\documentclass[11pt,a4paper]{article}

\usepackage{hyperref}

\usepackage{fullpage}

\usepackage{natbib}

\usepackage{pdflscape} % for landscape table page. 
\usepackage{enumerate}
\usepackage[usenames, dvipsnames]{xcolor} %for colored text
\usepackage{amsfonts} %for \mathbb
\usepackage{dsfont} %for \mathds
\usepackage{amsmath} %for \text
\usepackage{amsthm} %for \proof
\usepackage[final]{graphicx}
\usepackage{bm} %bold letters in math mode
\usepackage{ifthen} % for ifthenelse
\usepackage{mathtools} % modernized tweaks to amsmath
\usepackage{multirow}
\usepackage[section]{placeins} %floats stay in their section
\usepackage{float} %figure option H
\usepackage{url}

\usepackage{fancyhdr}
 
\RequirePackage{ifdraft}

\makeatletter
\def\parsetime#1#2#3#4#5\empty{#1#2:#3#4}
\def\parsedate#1:20#2#3#4#5#6#7#8\empty{20#2#3/#4#5/#6#7-\parsetime#8\empty}
\def\moddate#1{\expandafter\parsedate\pdffilemoddate{#1}\empty}
\newcommand{\timestamp}{%
\ifdraft{\textcolor{red}{\moddate{\jobname.tex}}}{\today}}
\makeatother

\makeatletter
\def\ifdraft{\ifdim\overfullrule>\z@
  \expandafter\@firstoftwo\else\expandafter\@secondoftwo\fi}
\makeatother

\newcommand{\draftremark}[1]{\ifdraft{%
   \textcolor{red}{// {#1} //}%
   \ifthenelse{\isodd{\value{page}}}{%
	\normalmarginpar\marginpar{\textcolor{red}{$\maltese$}}}{%
	\reversemarginpar\marginpar{\textcolor{red}{$\maltese$}}}}{}}

\newtheorem{theorem}{Theorem}[section]
\newtheorem*{theorem*}{Theorem}

\newtheorem{remark}[theorem]{Remark}

\newtheorem{example}[theorem]{Example}

\newtheorem{proposition}[theorem]{Proposition}
\newtheorem{properties}[theorem]{Properties}
\newtheorem{rem}[theorem]{Remark}
\newtheorem*{test*}{Test}
\numberwithin{theorem}{section}
\numberwithin{equation}{section}

\newcommand{\E}{\mathbb{E}}

\newcommand{\N}{\mathbb{N}}
\newcommand{\R}{\mathbb{R}}
\newcommand{\V}{\mathbb{V}}

\newcommand{\Prob}{\mathbb{P}}

\newcommand{\tMcor}{\overline{M\text{cor}}}
\newcommand{\cov}{\operatorname{cov}}
\newcommand{\cor}{\operatorname{cor}}
\newcommand{\dcor}{\operatorname{dcor}}
\newcommand{\dcov}{\operatorname{dcov}}

\newcommand{\ii}{\mathrm{i}}
\newcommand{\ee}{\mathrm{e}}

\newcommand{\xbar}{\overline{x}}
\newcommand{\ybar}{\overline{y}}
\newcommand{\Var}{\mathbb{V}}

\newcommand{\hN}{\mbox{}^{\scriptscriptstyle N}\kern-1.5pt}

\newcommand{\pairslr}[3]{\left#1 #2 \right#3}
\newcommand{\pairscs}[4]{{\csname#1l\endcsname #2} #3 {\csname#1r\endcsname #4}}

\newcommand{\pairs}[4][lr]{%
        \ifthenelse{\equal{#1}{}}{%
                #2 #3 #4}{%
                \ifthenelse{\equal{#1}{lr}}{\pairslr{#2}{#3}{#4}}{%
                        \pairscs{#1}{#2}{#3}{#4}}}}

\newcommand{\Expect}[2][lr]{\E\ifthenelse{\equal{#1}{lr}}{\!}{}\pairs[#1]{(}{#2}{)}}
\usepackage{mathtools}

\newcommand{\vapprox}[2]{\underset{\scriptstyle\overset{\mkern4mu\rotatebox{90}{$\,\approx$}}{#2}}{#1}}

\newcommand{\keywords}[1]{{\vspace{0.5cm}\noindent\bf Keywords:} #1}
\newcommand{\subclass}[1]{{\\\bf MSC 2010:} #1}

\begin{document}
\title{Notes on the interpretation of dependence measures\\[0.4cm]
\normalsize Pearson's correlation, distance correlation,\\ distance multicorrelations and their copula versions}
\author{Bj\"{o}rn B\"{o}ttcher\footnote{TU Dresden, Fakult\"at Mathematik, Institut f\"{u}r Mathematische Stochastik, 01062 Dresden, Germany, email: \href{mailto:bjoern.boettcher@tu-dresden.de}{bjoern.boettcher@tu-dresden.de}}}
\date{}
\maketitle

\begin{abstract}
Besides the classical distinction of correlation and dependence, many dependence measures bear further pitfalls in their application and interpretation. The aim of this paper is to raise and recall awareness of some of these limitations by explicitly discussing Pearson's correlation and the multivariate dependence measures: distance correlation, distance multicorrelations and their copula versions. The discussed aspects include types of dependence, bias of empirical measures, influence of marginal distributions and dimensions. 

In general it is recommended to use a proper dependence measure instead of Pearson's correlation. Moreover, a measure which is distribution-free (at least in some sense) can help to avoid certain systematic errors. Nevertheless, in a truly multivariate setting only the p-values of the corresponding independence tests provide always values with indubitable interpretation.

\subclass{62H20, %Multivariate analysis: Measures of association (correlation, canonical correlation, etc.) 
 62G05 %Non-parametric inference: Estimation 
}
\end{abstract}

\section{Introduction}

Most methods of statistical inference are in some way based on detecting and quantifying dependencies of variables. Hence for a rigorous analysis it is fundamental to understand the limitations of the dependence measures involved. We will discuss several aspects of Pearson's correlation as well as the more recent distance correlation of \cite{SzekRizzBaki2007} and its extensions to multiple variables: distance multicorrelation in \cite{BoetKellSchi2019} and \cite{Boet2020} and their copula versions in \cite{Boet2020a}. These measures provide a unifying concept which include as a special case also the bivariate Hilbert-Schmidt Independence Criterion of \cite{GretBousSmolScho2005} and as a limiting case the RV-coefficient of \cite{RobeEsco1976}. Besides many concrete examples which discuss the possibilities and limitations of these measures, also new theoretic results are included. In particular, a general extended invariance for copula dependence measures based on the distributional transform is proven. Moreover, the bias corrected estimators of distance covariance are extended to distance multivariance -- this is naturally important for any empirical application of these measures, but here it has to be taken with a pinch of salt: already for distance covariance these estimators turn out to have a much larger variance than the unbiased estimators. Thus without further developments (e.g.\ using variance reduction techniques) their use can not be recommended in a small sample setting.   

In the seminal paper of \cite{Reny1959} a set of axioms is given which a dependence measure for two univariate variables should satisfy. These axioms have been discussed, adapted and extended to many settings and measures. Our aim is not to challenge these (or other) axioms, but we aim to illustrate some of their practical implications. 

Dependence is a dichotomous concept: variables are either dependent or independent. This already shows that \textit{strength of dependence} is in general not a well defined concept. A measure that offers a quantitative scale might assign values to some range: from independence to some kind of specific non-random dependence, e.g.\ linear-dependence in the case of Pearson's correlation. Alternatively, the range might just be grounded at independence without any further reference values which would provide 'if and only if' characterizations of certain types of dependence. Nevertheless, in terms of speed (in a big data setting) it is certainly desirable to compare values of measures rather than requiring a (often computational involved) derivation of p-values. 
\begin{figure}[H]
\centering
\includegraphics[width = 0.49\textwidth]{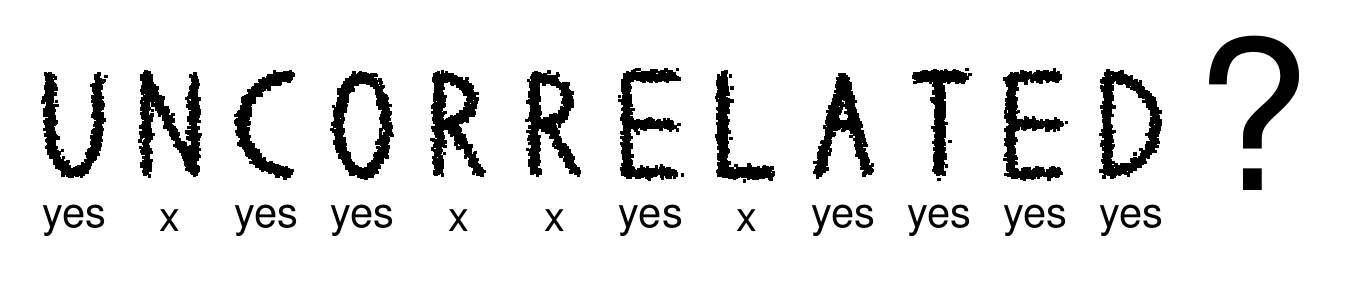}
\includegraphics[width = 0.49\textwidth]{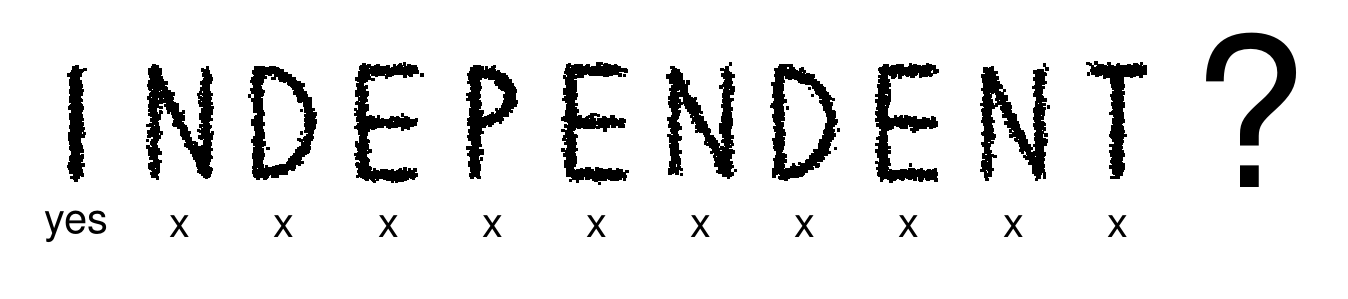}
\caption{Illustration of samples which are (un)correlated or (in)dependent - using a uniform distribution on a line representation of the letters perturbed by a bivariate normal distribution (with independent components). In this setting the only letter of the alphabet featuring independence is '\textsf{I}', but many feature uncorrelation. 
Note: A sufficient condition for uncorrelation is vertical or horizontal symmetry.\newline 
Conclusion: The use of proper dependence measures is essential to detect arbitrary dependence!}		\label{fig:uncorrelated-vs-independent}
\end{figure}
When using Pearson's correlation the most important limitation is the fact that it does not characterize independence, in many cases variables are uncorrelated but dependent, see e.g.\ Figure \ref{fig:uncorrelated-vs-independent}. Moreover, also large values appear for very different types of dependence as already illustrated in \cite{Ansc1973} (see also Figure \ref{fig:anscombe}). These aspects (with some extensions) will be discussed in Section \ref{sec:cov}. Recall that correlation is a measure for two univariate variables, in contrast distance multicorrelation is also applicable to multivariate settings.

In a truly multivariate setting all of the measures we discuss provide values which are not always suitable for direct comparisons without resorting to p-values. For distance correlation (and thus also for distance multicorrelation) a change of the marginal distributions can change the value of the measure systematically, see Example \ref{ex:changeofdistribution}. This can be overcome by the 'distribution-free' copula versions of these measures. But in a multivariate setting these still depend on the dependence within the components of a multivariate random variable under consideration, see Example \ref{ex:changeofdependence}. Thus also the latter do not always exclusively describe the dependence of the random variables under consideration.

Based on the following discussions the copula multicorrelations can be recommended in a setting where all variables are one-dimensional (i.e., $\R$-valued) with unknown continuous distributions. In all other cases (without further assumptions) each measure can yield systematic errors. In general, when considering random vectors only the p-values provide always a meaningful comparison, obviously in this case 'strength' and 'likelihood' would become synonymous.

In the next section we recollect properties of dependence measures. Thereafter we discuss Pearson's correlation (Section \ref{sec:cov}), distance correlation (Section \ref{sec:dcov}), distance multicorrelation (Section \ref{sec:Mcor}) and the copula versions of distance multicorrelation (Section \ref{sec:cop}). 
A conclusion is formulated in Section \ref{sec:conclusion}.

In all Figures \textbf{\textsf{cor}} denotes Pearson's correlation, \textbf{\textsf{Mcor}} denotes total distance multicorrelation based on the Euclidean distance (in the bivariate case it coincides with distance correlation) and \textbf{\textsf{CMcor}} denotes the copula version of total distance multicorrelation. For the latter measures we use bias corrected estimators and a sign preserving square root (i.e., $sign(x)\sqrt{|x|}$). Furthermore, the following notation is used throughout: $i \in\{1,\ldots,n\}$, $d_i \in \N$ and $X_i = (X_{i,1},\ldots,X_{i,d_1})$ with $X_{i,k}$ being univariate random variables. Then any dependence measure $d$ is a mapping of random variables $X_1,\ldots,X_n$ to some real-valued number, which satisfies (some of the) properties discussed in the next section. 

All simulations were executed based on the package 'multivariance' in the statistical computing environment R.

\section{Properties of dependence measures} \label{sec:prop}

A dependence measure can be multivariate in two ways: on the one hand it might consider more than two random variables at once (multiple variables, i.e., $n>2$), on the other hand the considered random variables might be random vectors (multivariate marginals; i.e., $d_i>1$). A measure is \emph{truly multivariate} if it can consider multiple random vectors (multiple multivariate marginals). Note, that it is a common abuse of terminology to call random vectors also (multivariate) random variables.

In general many aspects of a dependence measure might be of interest. We begin with a collection of commonly discussed properties.
\begin{properties}[Properties to classify dependence measures] \ \newline
\begin{description}
\item[\bf domain] A dependence measure might only be defined for a certain set of random variables, e.g., well-defined for all random variables with finite non-zero variance. (Depending on the size of the domain this property might also be called \emph{universality}, or simply \emph{existence}.)
\item[\bf range] The measure might only take values in a certain set, e.g., $[-1,1]$, $[0,1]$, $[0,\infty).$
\item[\bf known values] certain values might be meaningful, i.e., necessary and/or sufficient for a certain property. Of particular interest are: 
\begin{description}
\item[\bf characterization of independence] Some value of the measure characterizes the absence of (a certain type of) dependence, e.g., $d(X_1,\ldots,X_n)=0$ if and only if $X_i, i = 1,\ldots,n$ are independent (in this particular case the measure is also called a \emph{proper dependence measure}). 
\item[\bf characterization of types of dependence] Some value of the measure characterizes a certain type of dependence, e.g., $d(X_1,X_2)=1$ if $X_1$ is a linear function of $X_2$.
\item[\bf reference values] E.g., explicitly known values in the case of normal variates with known correlation, specifically this case might also be called \emph{Gaussian conformity}.
\end{description}
\item[\bf continuity] For sequences converging to some limit (within the domain!) the values of the measure converge to the value of the limit.
\item[\bf invariances] E.g.,
\begin{itemize}
\item {\bf permutation}-invariance:  $d(X_1,\dots,X_n)=d(X_{\pi_1},\ldots,X_{\pi_n})$ for all permutations $\pi_1,\ldots,\pi_n$ of $1,\ldots,n$. (In the case of two variables this is sometimes also called \emph{symmetry}.)
\item invariance with respect to \textbf{component-wise} transformations $S_i: \R^{d_i}\to \R^{d_i}$:\\ $d(X_1,\dots,X_n) = d(S_1(X_1),\ldots,S_n(X_n)).$ 
If the functions $x\mapsto S_i(x)$ are for each $i$ either $x$ or $-x$ then this properties is also called \emph{symmetry}.
\item invariance with respect to \textbf{elements-wise} transformations $g_{i,k}:\R \to \R$:\\ $d(X_1,\dots,X_n) = d(g_{1,1}(X_{1,1}), \ldots,g_{1,d_1}(X_{1,d_1}),\ldots,g_{n,1}(X_{n,1}), \ldots,g_{n,d_n}(X_{n,d_n}))$
\item invariance with respect to changes of the (univariate) \textbf{marginal distributions},\\ e.g. $d(X_1,\dots,X_n) = d(Y_1,\dots,Y_n)$ if $X_1,\ldots,X_n$ and $Y_1,\ldots,Y_n$ have the same copula. Variants of this property might also be called \emph{distribution-free}. 
\end{itemize}
\item[\bf metric-like] e.g.\ a triangle inquality holds when introducing further variables or when splitting a random vector into its components.
\item[\bf properties of corresponding sample versions] the sample versions might inherit the properties of the measure and have further specific properties, e.g., it can be a \emph{biased} or \emph{unbiased} estimator.
\end{description}
\end{properties}

The above provides an extensive list of common properties. We will focus in particular on the invariances and reference values, hereto the terms 'distribution-free', 'component-wise' and 'element-wise' are discussed in further details in the next Remarks. In Table \ref{tab:measures} the key properties of the measures considered in this paper are collected.

\begin{rem}[\emph{distribution-free} measures] \label{rem:distribution-free}
For univariate continuous marginal distributions a measure is distribution-free if and only if it is invariant with respect to strictly increasing (element-wise) transformations. 

For marginal distributions which are not continuous the copula is not unique (here a copula is a distribution function $C$ with uniformly distributed marginals such that $C(F_1,\ldots,F_l)$ is the joint distribution function it $F_k$ are the distribution functions of the univariate elements of the marginals). To get still a measure which some might call 'distribution-free', one can fix a selection procedure for a unique copula among all possible copulas (e.g., the corresponding linear extension copula, see Section \ref{sec:copversions}). But this selection procedure can introduce systematic errors when comparing continuous and non-continuous distributions.

In case of multivariate marginal distributions, a change of this multivariate distribution will in general change the value of the measure. To our knowledge there is no general way to get invariance with respect to arbitrary multivariate marginal distributions (dependence is characterized by copulas, and copulas join ultimately only  univariate distributions). Thus a truly multivariate dependence measure can only be 'distribution-free' with respect to changes of univariate distributions.
\end{rem}

\begin{rem}[component-wise vs.\ element-wise invariance]
For translations a distinction between component-wise and element-wise invariances is superficial, they coincide. As stated in the previous remark, distribution-free measures will feature further element-wise invariances. This is often desired, but certainly there are cases where a component-wise (scale and rotation) invariance seems also natural: For given position data of particles (in 2 or 3 dimensions), it is natural to assume an invariance of the dependence with respect to translations, rotations and scale of the underlying coordinate system. 
\end{rem}

\begin{landscape}
\begin{table}
\caption{Properties of $cor$, $Mcor$ and $CMcor$. Note for the 'known values': variables are related by a 'similarity transform' if distances based on one can be obtained from the other after transforming it by some combination of component-wise translation, scaling and rotation (for the explicit statement see \cite[Equation (38)]{Boet2020}). The moment condition required for total distance multicorrelation can be relaxed in specific cases, see \cite[Supplement, Proposition 1.1.(1)]{ChakZhan2019}.
}
\label{tab:measures}
{\small
\medskip
\centering 
\bgroup
\def\arraystretch{1.2}
\begin{tabular}{|l||l|l|l|l|}
	\hline
	                    & correlation                      & distance correlation        & total distance multicorrelations                                    & copula total distance multicorrelation                                 \\ \cline{2-5}
	measure             & $cor$                            & $Mcor$ for $n=2,\psi_i(.)=|.|$   & $\overline{Mcor},$ $\overline{Mcor}.lower,$ $\overline{Mcor}.upper$ & $C\overline{Mcor}$, $C\overline{Mcor}.lower,$ $C\overline{Mcor}.upper$ \\
	                    &                                  &     & $\overline{Mcor}.unnormalized$, ${Mcor}.pairwise$                   & $C\overline{Mcor}.unnormalized$, $C{Mcor}.pairwise$                    \\ \hline\hline
	domain              & univariate $X,Y$                 & multivariate $X,Y$               & random vectors $X_i, i = 1,\ldots,n$                                & random vectors $X_i, i = 1,\ldots,n$                                   \\
	                    & $\Var{X}<\infty, \Var{Y}<\infty$ & $\E{|X|}<\infty, \E{|Y|}<\infty$ & $\E\left(| \psi_i(X_i)|^n\right)<\infty$                            & -- {no condition} --                                                   \\ \hline
	range               & $[-1,1]$                         & $[0,1]$                          & $\overline{Mcor},\overline{Mcor}.lower,Mcor.pair.
	$: $[0,1]$;      & $C\overline{Mcor},$$C\overline{Mcor}.lower$,$CMcor.pair.$ : $[0,1]$    \\
	                    &                                  &                                  & $\overline{Mcor}.upper$,$\overline{Mcor}.unn.$: $[0,\infty)$        & $C\overline{Mcor}.upper$,$C\overline{Mcor}.unn.$: $[0,\infty)$         \\ \hline
	known values        & \phantom{-}0: \textit{uncorrelated}      & 0:  independence                 & 0:  independence                                                    & 0: independence                                                        \\
	                    &                                  &                                  & \phantom{0: }pairwise independence (for $Mcor.pair.$)               & \phantom{0: }pairwise independence (for $Mcor.pair.$)                  \\
	                    & \phantom{-}1: iff increasing line         & 1: iff similarity transform      & 1: if similarity transform                                          & 1: if identical                                                        \\
	                    & -1: iff decreasing line          &                                  & \ \ (for $Mcor.unn.$,$Mcor.pair.$ with $|.|^\alpha$;                & \ \ (modulo invariances; 'iff' for $Mcor.pair.$)                       \\
	                    &                                  &                                  & \ \ 'iff' for $Mcor.pairwise$)                                         &                                                                        \\ \hline
	continuity          & yes                              & yes                              & yes                                                                 & yes                                                                    \\
	(r.v.$\neq$ const.) &                                  &                                  &                                                                     &                                                                        \\ \hline
	invariances         & permutations                     & permutations                     & permutations                                                        & permutations                                                           \\
	(w.r.t.)            & translations                     & element-wise translations        & element-wise translations                                           & element-wise translations                                              \\
	                    & positive scalings                 & component-wise scalings               & component-wise scalings                                                  & element-wise scalings                                                   \\
	                    &                                  &                                  & \ \ (for $Mcor$,$Mcor.unn.$ with $|.|^\alpha$)                      &                                                                        \\
	                    &                                  & component-wise rotations              & component-wise rotations                                                 &                                                                        \\
	                    &                                  &                                  &                                                                     & element-wise monotone transformations                                             \\ \hline
	estimator           & biased                           & biased                           & biased                                                              & biased                                                                 \\ \hline
	conditions          & perturbed linear relation        & fixed marginal distributions     & fixed marginal distributions                                        & arbitrary marginals                                                    \\
	for sensible use    & between $X$ and $Y$              & \ \ (modulo invariances)         & \ \ (modulo invariances)                                            & \ \ (univariate and continuous)                                        \\ \hline
	comments            & no characterization              & systematic errors possible       & systematic errors possible                                          & systematic errors possible                                             \\
	                    & of independence                  & when varying marginals           & when varying marginals                                              & for non-continuous distributions                                       \\ \hline
\end{tabular}
\egroup
}
\end{table}
\end{landscape}

\section{Covariance and Pearson's correlation} \label{sec:cov}
Let $X$ and $Y$ be univariate random variables then covariance ($\cov$) and Pearson's correlation ($\cor$) are defined by
\begin{align}
\label{def:cov} \cov(X,Y) &= \E((X-\E(X))(Y-\E(Y))) = \E(XY) - \E(X)\E(Y),\\
\cor(X,Y) &= \frac{\cov(X,Y)}{\sqrt{\cov(X,X)\cov(Y,Y)}}. 
\end{align}
where correlation is well defined if the variances $\V(X)$ and $\V(Y)$ (i.e., $ \cov(X,X)$ and $\cov(Y,Y)$)  are finite and non-zero. The corresponding empirical estimator of $\cor(X,Y)$ for samples $(x^{(i)},y^{(i)})_{i =1,\ldots,N}$ of $(X,Y)$ is, denoting the mean by $\xbar  = \frac{1}{N}\sum_{i=1}^N x^{(i)}$, given by
\begin{equation} \label{eq:emp-cor}
\frac{\sum_{i=1}^N(x^{(i)}-\xbar)(y^{(i)}-\ybar)}{\sqrt{ \sum_{i=1}^N(x^{(i)}-\xbar)^2\ \sum_{i=1}^N(y^{(i)}-\ybar)^2}}.
\end{equation}
It is important to note that the above empirical correlation is a biased estimator for $\cor(X,Y)$, see e.g.\ \cite{ZimmZumbWill2003}. Thus in general one should be (at least a bit) careful when comparing values obtained based on varying sample sizes. 

In general, Pearson's correlation is a measure which attains values between -1 and 1. Only the values -1 and 1 have a clear interpretation: a perfect (increasing or decreasing) linear dependence of the variables. To avoid a common misconception, one also has to keep in mind the distinction of a linear relation of variables (i.e., the graph is a straight line) and a (perfect) linear dependence (i.e., the graph is a straight line which is neither horizontal nor vertical; see Figure \ref{fig:linear-relation}). For other values besides -1 and 1 correlation provides only a scale of numbers which is -- depending on the textbook -- subdivided into a scale from weak to strong (positive or negative) correlation. A value of 0 only indicates that the corresponding quantity is zero but this does not rule out any \textit{obivous} relation, as the examples in Figures \ref{fig:uncorrelated-vs-independent} and \ref{fig:classical-cor-dcor} illustrate. Moreover, a large value can appear for fundamentally different relations of the variables as the classical Anscombe-quartett, given in \cite{Ansc1973}, illustrates. In Figure \ref{fig:anscombe} we extend these classical examples to arbitrary sample sizes, to illustrate that this is a universal problem and not just related to small sample sizes.

A mathematical clear (but maybe not always useful) interpretation of correlation appears in a regression model. Suppose one tries to predict $Y$ by a linear function of $X$ (i.e., $a+bX$) using a least squares approximation. Then $cor^2(X,Y)$ is the proportion of the variance of $Y$ which is explained by the model. Note that in settings appropriate for regressions, i.e., perturbed linear relations of the variables, simulations indicate that the values of $Mcor$ and $CMcor$ behaves similar to $|cor|$, see Figure \ref{fig:linear-relation-perturbed}. 

In other contexts correlation provides just a scale from $-1$ to $1$ to which some meaning is imputed by reference values: the standard reference is the multivariate normal distribution with unit variances, where the correlation appears explicitly as a parameter in the definition of the distribution. But this reference is very deceiving, since for a multivariate normal distribution a zero correlation of components is equivalent to their independence. Thus we want to emphasize once more: In general a zero correlation says nothing about the presence of dependence, as it was illustrated e.g.\ in Figures \ref{fig:uncorrelated-vs-independent} and \ref{fig:classical-cor-dcor}. Hence, Pearson's correlation is not a proper dependence measure.   

\begin{figure}[H]
\centering
\includegraphics[width = 0.99\textwidth]{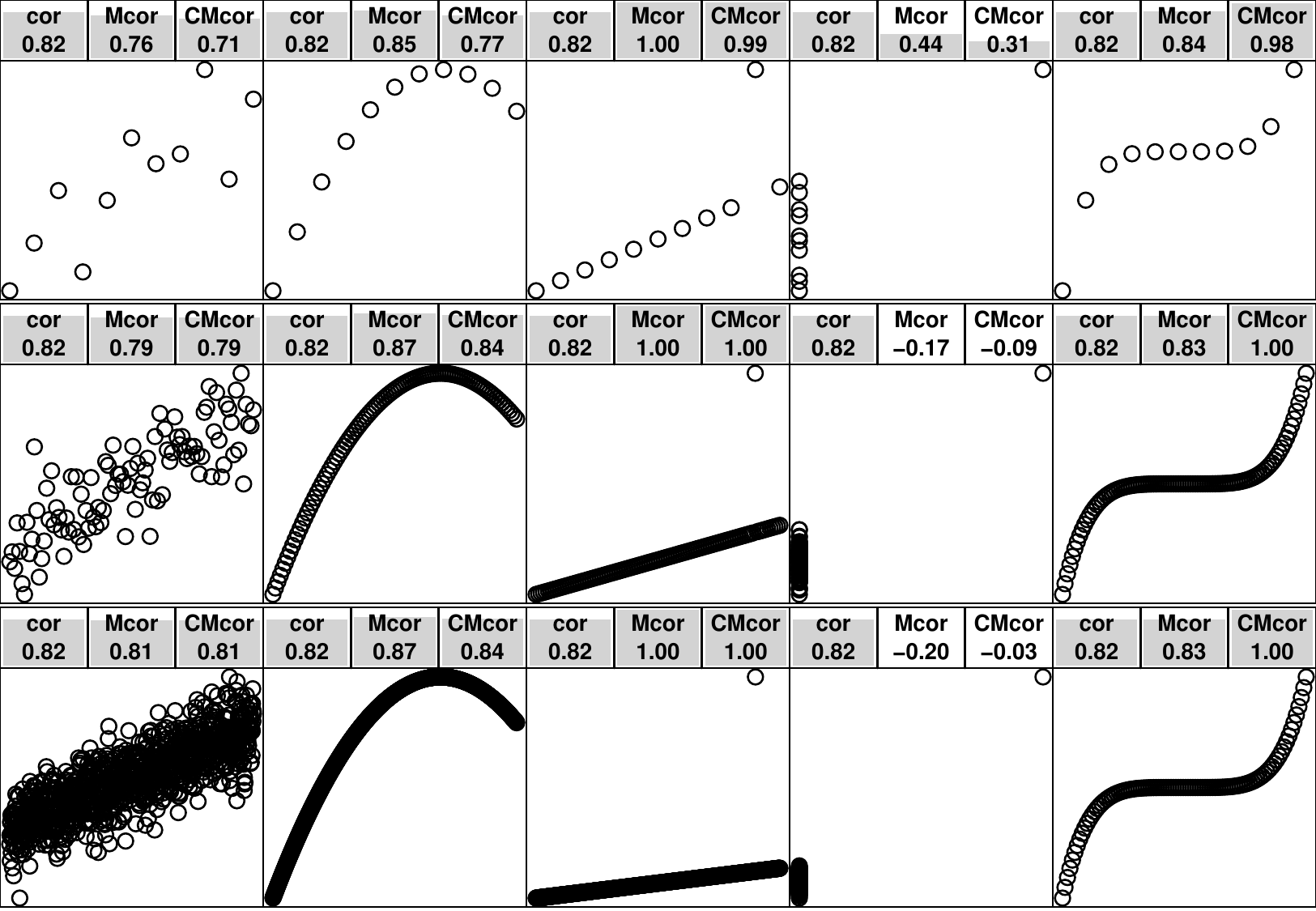}
\caption{Extended Anscombe's Quartett - the first four plots in the first row are the classical Anscombe's Quartett of \cite{Ansc1973}. We added a strictly monotone transformation as fifth column and extend all examples from the original sample size 11 (first row) to arbitrary sample sizes ($N = 100$ and 1000 are depicted; the code is provided in the online supplement). Key observation: large values of correlation ($\cor$) do not indicate a specific type of dependence. $Mcor$ and $CMcor$ quantify these dependencies differently. The third and forth column indicate also a robust behavior of these multicorrelations.}		\label{fig:anscombe}
\end{figure}

\section{Distance multivariance and its multicorrelations} \label{sec:Mcor}

For a comprehensive introduction to the framework of distance multivariance see \cite{Boet2020}. Here we focus on total distance multicorrelation, which characterizes (by a value of 0) independence directly. The special case of distance correlation, which considers only two variables instead of arbitrary many, is discussed in Section \ref{sec:dcov}. The copula variants of the measures are discussed in Section \ref{sec:copversions}.

Fundamental to the setup are so called 'real-valued continuous negative definite' functions $\psi_i: \R^{d_i} \to \R$ (in the sense of \cite{BergFors75}) for which the standard choice is the Euclidean distance $|.|$ (for many other options see e.g.\ \cite[Table 1]{BoetKellSchi2018}). 
Given $\psi_i$ there are several ways to define induced measures in the framework of distance multivariance. We will use here a representation based on expectations, more involved definitions using characteristic functions and for certain cases computational more efficient formulas do exist, see \cite{Boet2020} and the variants discussed in Section \ref{sec:multi}. Total distance multicorrelation $\overline{Mcor}$ is defined by\footnote{Here we use for simplicity a different naming convention than in \cite{Boet2020}.}
\begin{align}\label{eq:deftMcor}
\overline{Mcor}(X_1,\ldots,X_n)&: = \sqrt{\frac{1}{2^n-n-1} \sum_{\substack{1\leq i_1<\ldots < i_m \leq n\\ 2\leq m\leq n}} Mcor^2(X_{i_1},\ldots,X_{i_m})}\\
\label{eq:defMcor} &\text{with }Mcor^2(X_{i_1},\ldots,X_{i_m}) = {\E\left(\prod_{l=1}^m \frac{-\mathcal{C}_{X_{i_l}}\mathcal{C}_{X'_{i_l}}\psi_{i_l}(X_{i_l}-X'_{i_l})}{c_{i_l,m}}\right)}
\end{align}
where $(X_1',\ldots,X_n')$ is an independent copy of $(X_1,\ldots,X_n)$, $\mathcal{C}_Y Z:= Z - \E(Z\mid Y),$ and  $c_{i,m} := \E(|\mathcal{C}_{X_{i}}\mathcal{C}_{X'_{i}}\psi_{i}(X_{i}-X'_{i})|^m)^\frac{1}{m}.$ One might remove the square root in \eqref{eq:deftMcor} by considering the square of $\overline{Mcor}$, here we decided to keep the root in order to be on the same scale as Pearson's correlation, see Remark \ref{rem:Mcor-cor}. A sufficient condition for the finiteness of \eqref{eq:deftMcor} is given in Table \ref{tab:measures}.

For samples $\bm{x}^{(k)} = (x_1^{(k)},\ldots,x_n^{(k)}), k=1,\ldots,N$ of $(X_1,\ldots,X_n)$ let $D_i:=(\psi_i(x_i^{(j)}-x_i^{(k)}))_{j,k=1,\ldots,N}$ be the distance matrices of component $i= 1,\ldots,n$. Then an estimator for $Mcor^2(X_{i_1},\ldots,X_{i_m})$ is
\begin{equation} \label{eq:sampleMcor}
{\frac{1}{N^{2}} \sum_{j,k=1}^N \prod_{l=1}^m  \frac{-(CD_{i_l}C)_{j,k}}{\hat c_{i_l,m}}}
\end{equation}
with $\hat c_{i_l,m} = \sqrt[m]{\frac{1}{N^{2}} \sum_{j,k=1}^N  |(CD_{i_l}C)_{j,k}|^m},$ where $C$ is the $N\times N$ centering matrix, i.e., it has elements $C_{j,k} = \delta_{j,k}-\frac{1}{N}.$ A combination of these estimators yields an estimator for \eqref{eq:deftMcor}, which is due to the square root in \eqref{eq:deftMcor} always biased. But at least for the estimators of the individual terms one can use a bias correction.

\begin{remark}[Bias and bias correction]
The estimator introduced above for $Mcor^2$ is biased, which has two main causes: 1. the constants are also estimated, and it is (as it is well known for correlation) usually not (easily) possible to construct unbiased estimators for quotients. 2. also in the case that the constants are fixed, the estimator for the numerator is biased. 

The latter can be improved by using a \emph{bias corrected} estimator, hereto $(CD_iC)_{j,k}$ in \eqref{eq:sampleMcor} is replaced by 0 for $j=k$ and otherwise by
\begin{equation}\label{eq:biascorrection}
\begin{split}
\psi_i(x_i^{(j)}-x_i^{(k)}) &- \frac{1}{n-2} \sum_{l=1}^N \psi_i(x_i^{(l)}-x_i^{(k)})\\
&- \frac{1}{n-2} \sum_{m=1}^N \psi_i(x_i^{(j)}-x_i^{(m)}) + \frac{1}{(n-1)(n-2)} \sum_{l,m=1}^N \psi_i(x_i^{(l)}-x_i^{(m)})
\end{split}
\end{equation}
and the factor $\frac{1}{N^2}$ in \eqref{eq:sampleMcor} is replace by $\frac{1}{N(N-3)}$ (requiring $N>3$), see \cite[Definition 2]{SzekRizz2014} and \cite[Section 3.2]{ChakZhan2019}, which translate in this context directly to the general setting of distance multivariance (replacing $|.|$ with $\psi_i(.)$). 

Note, the estimator is purposely called \emph{bias corrected} instead of \emph{unbiased}, since it is only unbiased if $n = 2$ and $c_i$ is a fixed constant. For $n>2$ there is still some \textit{higher order} bias (Figures \ref{fig:bias-multi-1}-\ref{fig:bias-multi-3}). This has been discussed (for a specific $c_m$) in \cite[Section 3.2]{ChakZhan2019}. Moreover, if the constants $c_i$ are replaced by estimators the whole object is again (as in the case of Pearson's correlation) not an unbiased estimator.

For the purpose of a direct comparison of values one usually would like to be able to compare these also for varying sample sizes, thus it is generally recommended to use bias corrected versions. But here it has two notable (and usually undesired) side-effects: 1. the range of the estimator is larger than the range of the estimated measure, since the estimator can take negative values (this must happen, since otherwise an unbiased estimate of 0 would not be possible). 2. the bias corrected estimator has for small samples a much larger variance than the biased estimator (thus values become unreliable), see Figures \ref{fig:bias-2}, \ref{fig:bias-multi-1} and \ref{fig:bias-multi-2}.
\end{remark}

\begin{figure}[H]
\centering
\includegraphics[width = 1\textwidth]{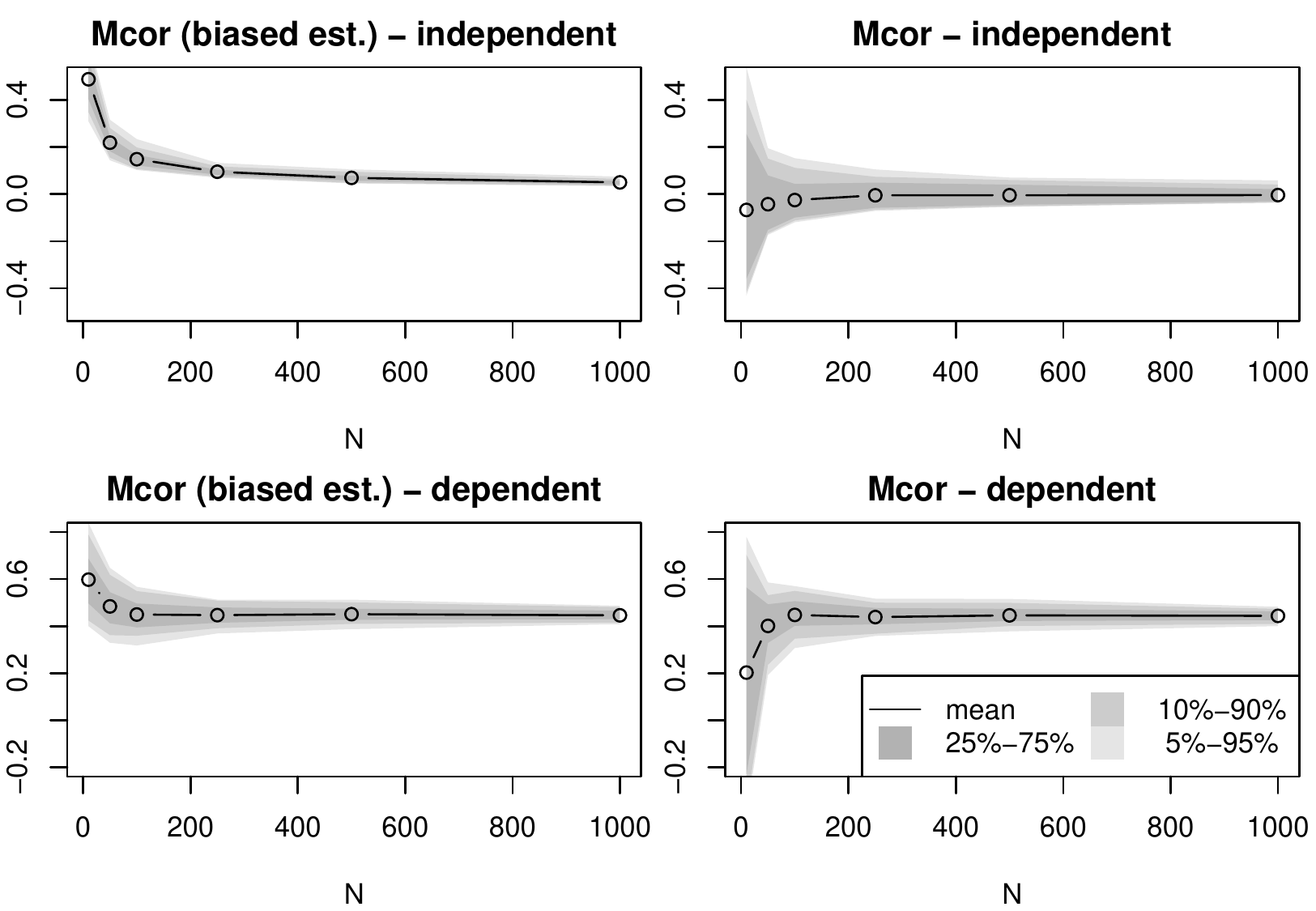}
\caption{Comparison of biased and bias corrected estimates. The data was sampled with sample size $N$ from a pair of uniformly distributed random variables which are either independent (first row) or coupled by a Gaussian copula with correlation 0.5 (second row). Observation: It is clearly visible that the bias corrected estimate have less (or different) bias, and they feature a larger variance.}		\label{fig:bias-2}
\end{figure}

\begin{remark}[Direct link of correlation and distance multicorrelation] \label{rem:Mcor-cor} Distance multivariance is not only an alternative to Pearson's correlation, in fact, it is a unifying theory which includes Pearson's correlation as a limiting case.
As discussed in \cite[Section 3.1]{Boet2020} Pearson's correlation appears as a special limiting case of distance multivariance: if $Mcor_\alpha$ is the multicorrelation corresponding to $\psi_i(.) = |.|^\alpha$ and $X$,$Y$ are univariate then
\begin{equation}
\lim_{\alpha \nearrow 2}Mcor_{\alpha}(X,Y) = |Cor(X,Y)|.
\end{equation}
But note that the limit $|.|^2$ of $\psi_i$ is not a valid function for distance multicorrelation, since there is no corresponding L\'evy measure with full support (cf.\ \cite{Boet2020}). Therefore this statement does not contradict the proper characterization of independence by distance multicorrelation. 
\end{remark}

\subsection{The case of two random variables: Distance covariance and distance correlation} \label{sec:dcov}

The special case of distance multivariance with two variables is also known as generalized distance covariance introduced in \cite{BoetKellSchi2018}. In particular for $\psi(.)=|.|$ (or $|.|^\alpha$ with $\alpha \in (0,2)$) it becomes distance covariance and distance correlation ($\dcor$) of \cite{SzekRizzBaki2007}:
\begin{align}
\tMcor(X,Y) = Mcor(X,Y) = \dcor(X,Y). 
\end{align}
Distance correlation has also the more elementary representation 
\begin{align}
\label{eq:dcor} &\dcor(X,Y) = \frac{\dcov(X,Y)}{\sqrt{\dcov(X,X) \dcov(Y,Y)}},\\
\text{where }&\dcov(X,Y) = \sqrt{\int |f_{(X,Y)}(s,t) - f_X(s)f_Y(t)|^2 \, dsdt} \\
&\phantom{\dcov(X,Y) }= \sqrt{\int |\cov(\ee^{\ii X\cdot s},\ee^{\ii Y\cdot t})|^2\, dsdt}.
\end{align}
The measure $\dcor$ is invariant with respect to rotations, translations and component-wise scalings (with a joint non-zero factor for the whole component), see \cite{MoriSzek2018}. But its values depend on the specific marginal distribution as the following examples illustrate.

\begin{example}[Influence of marginal distributions] \label{ex:changeofdistribution}
Let $(X_{norm},Y_{norm})$ be normal distributed with means 0, variances 1 and covariance $\rho$, and set $X_{unif}  = F_{norm}(X_{norm})$ and $Y_{unif} = F_{norm}(Y_{unif})$, where $F_{norm}$ is the distribution function of the standard normal distribution. Then $X_{unif}$ and $Y_{unif}$ have as distribution function the so called \emph{Gaussian copula} with correlation $\rho$.
The subscripts $exp,\ {chi},\ {bern}$ denote the corresponding random variables obtained from the uniformly distributed random variables via the (generalized) inverse distribution function of the exponential distribution with parameter 1, the Chi-squared distribution with parameter 1 and the Bernoulli distribution with paramter 0.5, respectively. Thus these random variables are exponential, Chi-squared and Bernoulli distributed with the same Gaussian copula. 

For $\rho = 0$, i.e., in the case of independence the population measure $Mcor$ becomes 0, and one might expect that the sample measures all indicate this in the same manner. But it turns out, that already in the case of independent variables the relative size of the sample measure depends in general on the prescribed marginal distributions when using the biased estimators (Figures \ref{fig:pair-freq} and \ref{fig:pair-freq-full}). For the unbiased estimators this systematic problem occurs only for $\rho>0$ (Figures \ref{fig:pair-freq-unbiased} and \ref{fig:pair-freq-unbiased-full}).

The figures only indicate an existence of a relative size difference, to illustrate this by explicit values we state some of these for the case $\rho = 0.8$ (including Bernoulli marginals which are not in the figures). The values are the means of 100 cases based on 1000 samples each: 
\begin{equation*} 
\begin{split}
\vapprox{Mcor(X_{exp},Y_{chi})}{0.73} > \vapprox{Mcor(X_{norm},Y_{chi})}{0.70}&> \vapprox{Mcor(X_{unif},Y_{chi})}{0.69}\\
&> \vapprox{Mcor(X_{exp},Y_{bern})}{0.61} > \vapprox{Mcor(X_{chi},Y_{bern})}{0.58}
\end{split}
\end{equation*}
and
\begin{equation*}
\vapprox{Mcor(X_{norm},Y_{norm})}{0.75} > \vapprox{Mcor(X_{unif},Y_{unif})}{0.74} > \vapprox{Mcor(X_{chi},Y_{chi})}{0.73} > \vapprox{Mcor(X_{bern},Y_{bern})}{0.58}. 
\end{equation*}
The figures and the inequalities show explicitly that systematic errors would occur if these dependencies are classified/ranked/ordered by comparing the values of the measure.
\end{example}

\begin{figure}[H]
\newlength{\myl}
\setlength{\myl}{0.47\textwidth}
\centering
\includegraphics[width = \myl]{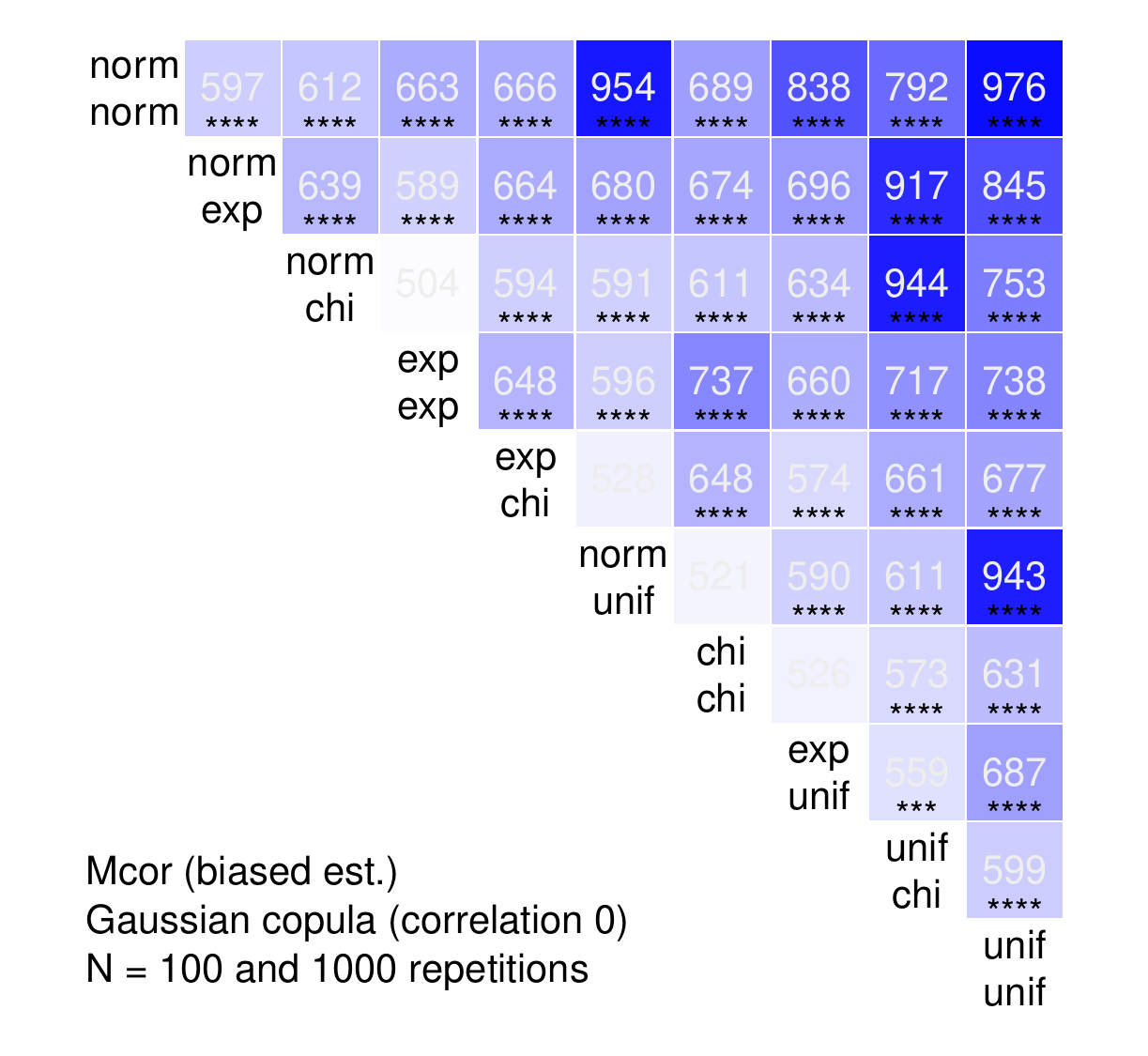}
\includegraphics[width = \myl]{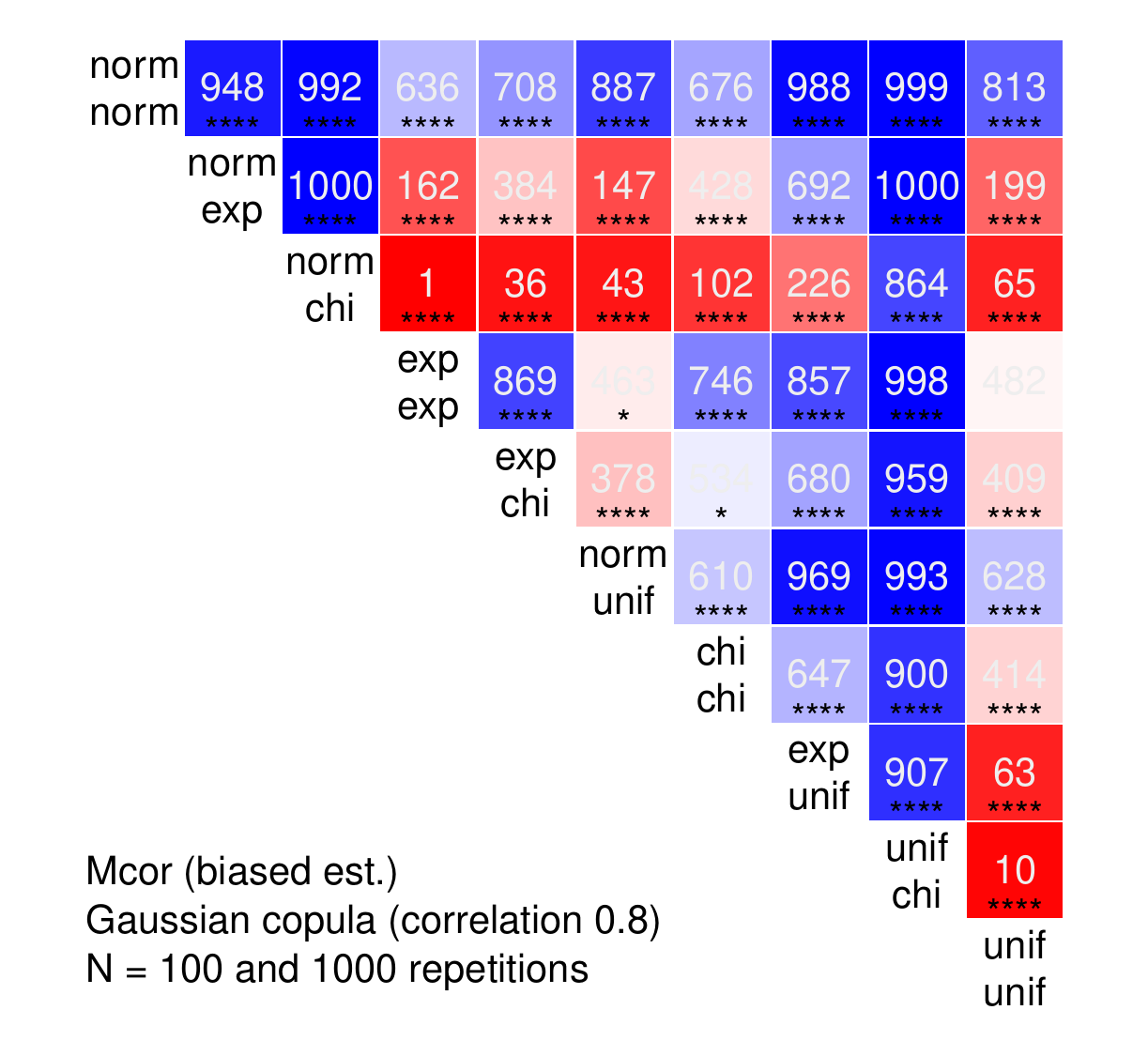}

\caption{Systematic dominance due to marginal distributions. The Figure illustrates how often $Mcor(X_{row},Y_{row})$ $>$ $Mcor(X_{column},Y_{column})$ out of 1000 simulated cases. E.g., the top right corner of the left graph illustrates for $X_{norm},Y_{norm}$ being independent normally distributed random variables and $X_{unif},Y_{unif}$ being independent uniformly distributed random variables that the sample estimate (based on a sample of size 100; using the biased estimator) of $Mcor(X_{norm},Y_{norm})$ was in 976 out of 1000 cases larger then the sample estimate of $Mcor(X_{unif},Y_{unif})$. The marginals are connected by a Gaussian copula with correlations 0 and 0.8, respectively for the left and right graph (see Figure \ref{fig:pair-freq-full} for further values). The labels indicate which marginal distributions are used. The expected number of dominances of one pair over another (if the value of the measure would not depend on the marginal distributions) is 500. The true number of dominances is printed together with the corresponding p-value in 'star notation' (*: $\leq$0.05, **: $\leq$0.01, ***: $\leq$0.001, ****: $\leq$0.0001; after adjusting for multiple tests by Holm's method). The background of the number becomes more opaque with increasing deviation from the expected value (red for more, blue for less). Observation: One clearly observes a systematic dependence on the marginal distribution. Especially (and maybe surprisingly), this is also the case for independent variables. Compare with Figures \ref{fig:pair-freq-unbiased} and \ref{fig:pair-freq-cop}.}		\label{fig:pair-freq}
\end{figure}

\begin{figure}[H]
%\newlength{\myl}
\setlength{\myl}{0.47\textwidth}
\centering
\includegraphics[width = \myl]{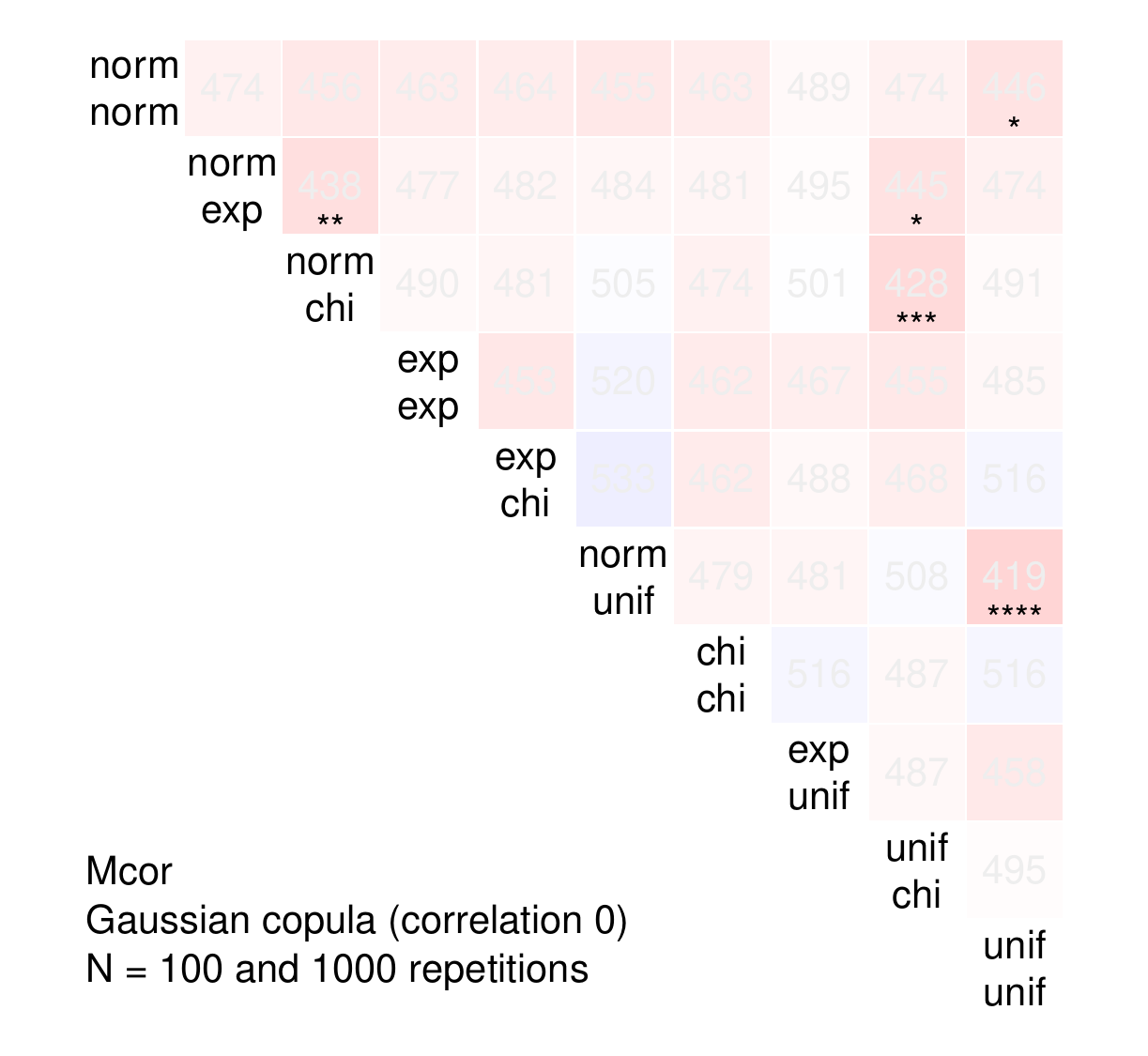}
\includegraphics[width = \myl]{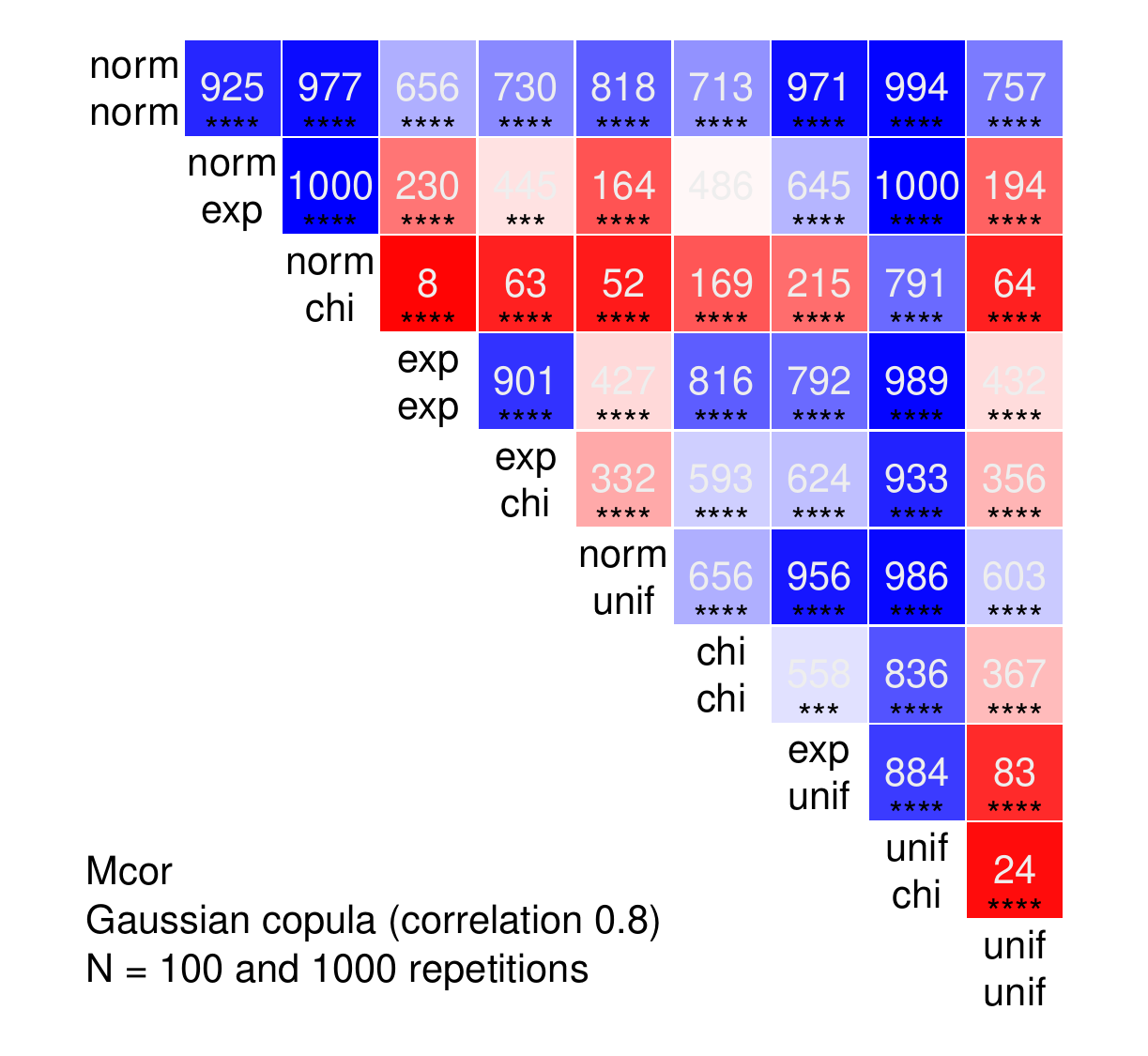}

\caption{Systematic dominance due to marginal distributions. Same as Figure \ref{fig:pair-freq} but using the bias corrected estimators, see Figure \ref{fig:pair-freq-unbiased-full} for additional parameter settings. Observation: There is in general a systematic dependence on the marginal distributions as the right graph indicates. But in the special case of independent variables this systematic problem almost disappears.}		\label{fig:pair-freq-unbiased}
\end{figure}

\begin{example}[Influence of the dependence within a margin] \label{ex:changeofdependence}
For $s\in [0,1]$ let the 3-dimensional random vector $(X_{1,norm}^{s},X_{2,norm}^{s},Y_{norm})$ be normally distributed with zero means, variances 1 and covariances $\cov(X_{1,norm}^{s},X_{2,norm}^{s}) = s$ and $\cov(X_{i,norm}^{s},Y_{norm}) = 0.5$ for $i = 1,2$. As in Example \ref{ex:changeofdistribution} the corresponding uniformly distributed random variables are denoted by the subscript "unif".  The values are the means of 100 cases based on 1000 samples each: 
\begin{equation}
\begin{split}
\vapprox{Mcor\left((X^{0}_{1,unif},X^{0}_{2,unif}),Y_{unif}\right)}{0.56} &> \vapprox{Mcor\left((X^{0.5}_{1,unif},X^{0.5}_{2,unif}),Y_{unif}\right)}{0.52} \\&> \vapprox{Mcor\left((X^{1}_{1,unif},X^{1}_{2,unif}),Y_{unif}\right)}{0.46}.
\end{split}
\end{equation}
Hence also a dependence within a random vector can systematically change the value of the measure.
\end{example}

\subsection{The case of multiple random variables} \label{sec:multi}

Multivariate dependence measures are important because they allow to (quickly) analyze many variables by global tests. Moreover, only these measures can detect higher order dependencies directly. From our point of view such testing approaches are very valuable. But to use these measures to quantify dependence directly (without testing) seems less sensible, since it summarizes many (possibly very complex) relations into a single number, making a meaningful interpretation almost impossible (see e.g.\ Figure \ref{fig:dep-structures}). Nevertheless, it provides a scale of dependence and hence we will discuss some aspects.  

\begin{figure}[H]
%\newlength{\myl}
\setlength{\myl}{0.9\textwidth}
\centering
\includegraphics[width = \myl]{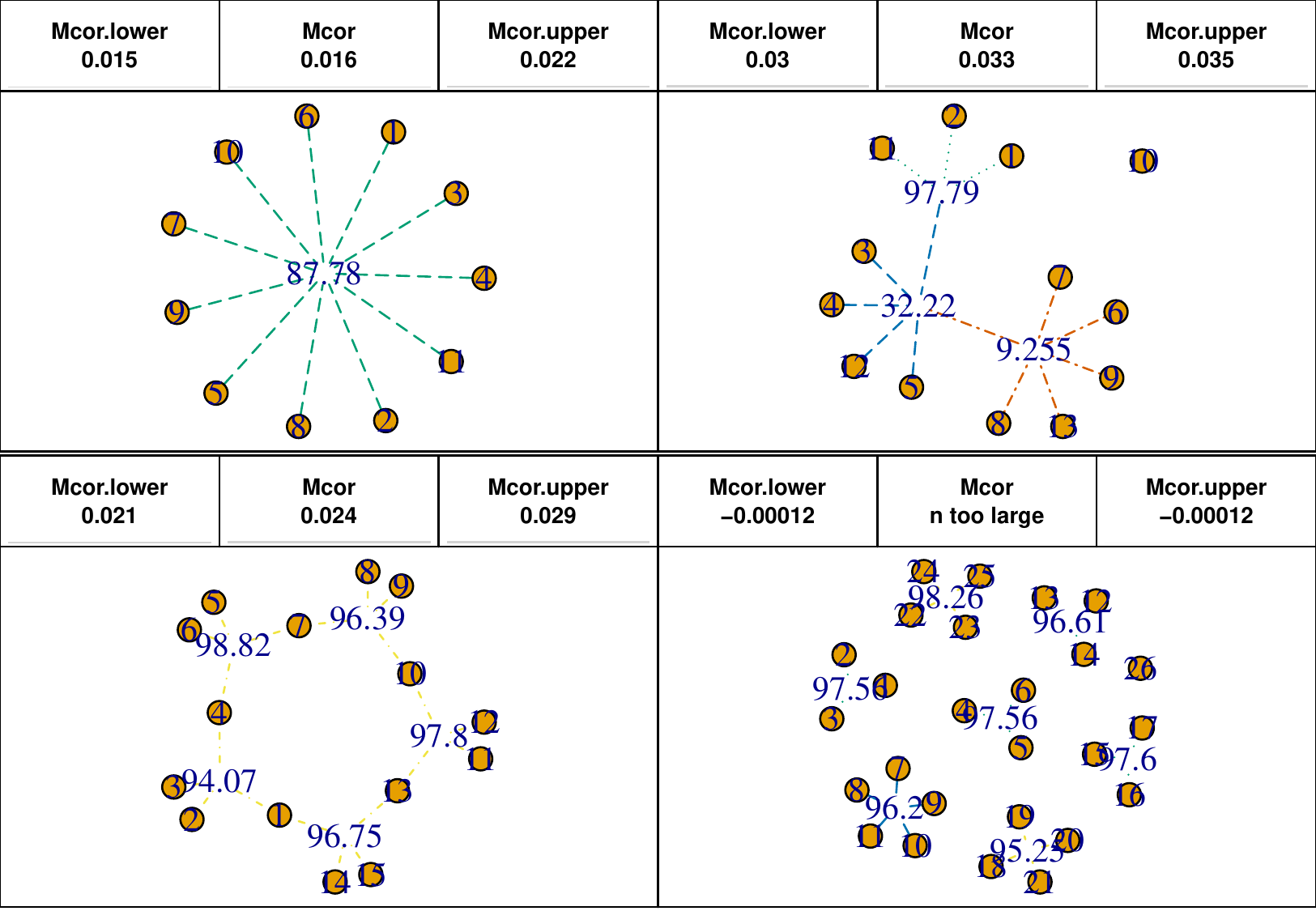}

\caption{Higher order dependence structures. Each figure shows an estimated (based on 100 samples) dependence structure and the corresponding multicorrelations. The examples are constructed explicitly in \cite[Examples 10.3,10.5,10.7,10.8]{Boet2020}, wherein also the estimation procedure is developed. Observation: The values of multicorrelation are not appropriate to distinguish the examples. In contrast to these values the (estimated) illustrations of the structure contain much more information. The values in the dependence structures denote the value of the test statistic of the connected variables (see \cite{Boet2020} for details).}		\label{fig:dep-structures}
\end{figure}

The examples of the previous section (Examples \ref{ex:changeofdistribution}, \ref{ex:changeofdependence} and Figure \ref{fig:bias-2}) indicate the limitations of the measures in the bivariate setting. These persist (obviously) also if more than two variables are considered, see e.g.\ Figures \ref{fig:bias-multi-1} and \ref{fig:bias-multi-2} for the comparison of biased and bias corrected estimators in the multivariate setting. Additionally, further difficulties arise in this setting, because properties which are desirable (range $[0,1]$; meaningful value of 1; computational fast estimator for large $n$) do actually require different extensions of the bivariate case. All are obtained by using different constants $c_{i,m}$ in \eqref{eq:defMcor}, which have in common that the resulting measures are component-wise scale invariant if $\psi_i(x_i)=|x_i|^\alpha$ for $\alpha\in(0,2)$, and all theses measure coincide for $n=2.$ 
\begin{itemize}
\item $\overline{Mcor}$ -- total multicorrelation -- was introduced in \eqref{eq:deftMcor}. It has range $[0,1]$. Its estimator is not suitable for large $n$ due to the many summands in \eqref{eq:deftMcor}. 
\item $\overline{Mcor}.lower$ and  $\overline{Mcor}.upper$ -- lower and upper bound for total multicorrelation --  use in \eqref{eq:defMcor} the constants $c_{i,m} := \E(|\Psi_i(X_i,X_i')|^2)^\frac{1}{2}$ and $c_{i,m} := \E(|\Psi_i(X_i,X_i')|^n)^\frac{1}{n},$ respectively. The lower bound was used in \cite{ChakZhan2019} and the upper bound in \cite{Boet2020}. Since the constants do not depend on $m$ these measures can be computed more efficiently, the combined estimator becomes (with the notation of \eqref{eq:sampleMcor})
\begin{equation} \label{eq:sampletMcor}
\sqrt{\frac{1}{2^n-n-1} \left(\left(\frac{1}{N^{2}} \sum_{j,k=1}^N \prod_{i=1}^n (1 - \frac{(CD_iC)_{j,k}}{c_i})\right) -1\right)}.
\end{equation}
Note that, when using the bias corrected versions (see \eqref{eq:biascorrection}) the values may actually not be 'upper' and 'lower' bounds, since in this case not only positive terms but also negative terms (which are due to the bias correction) get the same scaling.
\item $\overline{Mcor}.unnormalized$ -- unnormalized total multicorrelation -- is obtained using in \eqref{eq:defMcor} the constant
\begin{equation}
c_{i,m} := \E((\Psi_i(X_i,X_i'))^m)^\frac{1}{m} =  (M(\underbrace{X_i,\dots,X_i}_{m-times}))^\frac{1}{m}. 
\end{equation} 
This seems a very natural choice considering the bivariate representation in \eqref{eq:dcor}, and one can show that a value of 1 appears if the random variables are related by similarity transforms, see e.g.\ \cite[Section 3.6]{Boet2020}. But the measure has a possibly unbounded range and the value 1 can also occur for other cases, see Figure \ref{fig:multivariate}. Moreover, there is no computational feasible estimator for large $n$ and the estimate of the norming constant can become zero for non-constant random variables (e.g.\ for a Bernolli component with exactly $N/2$ successes and $m$ odd). Furthermore, as a consequence its estimator does not always converge properly (Figure \ref{fig:bias-multi-3}).
\item $Mcor.pairwise$ -- 2-multicorrelation or pairwise multicorrelation -- was introduced in \cite[Equation (56)]{Boet2020} and it is obtained by fixing $m=2$ in \eqref{eq:deftMcor}, i.e., it is based exclusively on the multicorrelations of all pairs. Hence it is not a proper dependence measure, but it characterizes pairwise independence, has bounded range [0,1] and the value 1 occurs if and only if all variables are related by similarity transforms. 
\item $\overline{\mathcal{M}}$ -- normalized total multivariance -- is obtained using in \eqref{eq:defMcor} the constant $c_{i,m} = \E(\psi_i(X_i-X_i')).$ This yields the test statistic used in independence tests based on distance multivariance, nevertheless it also assigns values to dependence and might be considered as a dependence measure. The estimator of $\overline{\mathcal{M}}^2$ scaled by $N$ is non-negative and has under the assumption of independence (for non-constant random variables) unit expectation. Large values (in comparison to its expectation 1) indicate dependence, and (in case of dependence) diverge with increasing sample size. The values can be translated 1-to-1 to (possibly very) conservative p-values. Thus at least for fixed sample size they provide a (very) rough comparison. Its building blocks were used in \cite{Boet2020} (see also Figure \ref{fig:dep-structures}) to provide some basic quantification of the dependencies for detected higher order dependence structures (where the required values are readily available from the detection algorithm).  
\end{itemize}
Note, further related measures for special cases exist: e.g.\ if the variables are known to be lower-order independent then multicorrelation instead of total multicorrelation can be used to quantify the dependence (see \cite{Boet2020}). Such measures might be called conditional dependence measures in contrast to the (proper) dependence measures discussed here.

The lower and upper bound ($\overline{Mcor}.lower$ and  $\overline{Mcor}.upper$) are the only appropriate choice for a proper dependence measure in the case of many variables (large $n$), considering the computational expense of the estimators. But in general each measure has its advantages. In the next example the performance of the above measures is depicted for several settings. 

\begin{example}[Values of distance multicorrelations in a multivariate setting]\label{ex:multivariate} We consider the following three cases:
\begin{enumerate}
\item multivariate normal: The random vector $(X_1,X_2,X_3)$ is multivariate normally distributed with mean 0, variances 1 and pairwise covariances $s\in [0,1]$.
\item linear interpolation to complete dependence: Let $X,X_1,X_2,X_3$ be independent standard normal random variables then consider the random vector $s(X,X,X)+(1-s)(X_1,X_2,X_3)$ for $s\in [0,1]$. 
\item perturbed higher order dependence: Let $X_1,X_2,X_3$ be independent standard normal random variables, $Y_1,Y_2$ be independent Bernoulli variables and $Y_3$ is 1 if $Y_1$ and $Y_2$ have the same value, otherwise $Y_3$ is 0. Then $Y_1,Y_2,Y_3$ are dependent but pairwise independent (cf.\ \cite[Example 10.2]{Boet2020}). Now consider
$(Y_1,Y_2,Y_3) + s (X_1,X_2,X_3)$ for $s\in[0,1].$
\end{enumerate} 
The values of the measures for these examples are shown in Figure \ref{fig:multivariate}, based on a sample of size 100 for each setting.
\end{example} 

 \begin{figure}[H]
 %\newlength{\myl}
 \setlength{\myl}{0.99\textwidth}
 \centering
 \includegraphics[width = \myl]{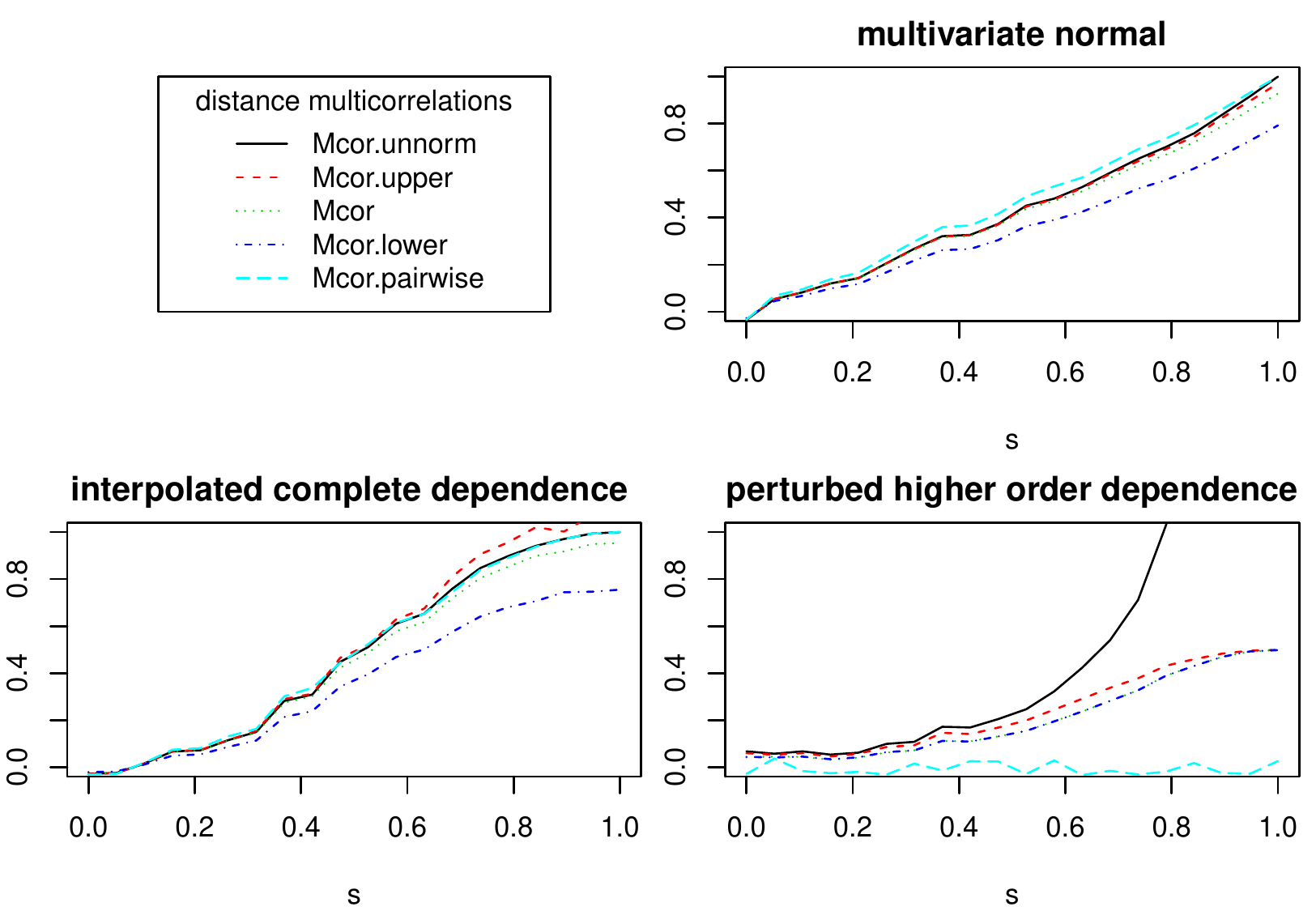}
 
 \caption{Illustrations of the values of distance multicorrelations for Example \ref{ex:multivariate}. Observation: $\overline{Mcor}.lower$ and $\overline{Mcor}.upper$ provide bounds for $\overline{Mcor}$, and $\overline{Mcor}.unnormalized$ becomes larger than 1 in the case of higher order dependence.}		\label{fig:multivariate}
 \end{figure}

\section{Copula version of distance multicorrelation} \label{sec:cop}

As discussed in the previous section the value of distance multicorrelation depends on the marginal distributions. In this section we present a version which only depends on an underlying coupla and it is therefore (for continuous marginals) not influenced by the univariate marginal distributions (see also Remark \ref{rem:distribution-free}). Hereto we recall and extend the results of \cite{Boet2020a}, which are based on the distributional transform discussed by \cite{Nes2007} and \cite{Rues2009}.

\subsection{distributional transform}

For a univariate random variable $X$ define 
\begin{equation}
T_X(x,u):= \Prob(X< x) + u \Prob(X=x) \text{ for all }x\in\R, u\in[0,1].
\end{equation} 
Let $U$ be a uniformly distributed random variable independent of $X$, then 
\begin{equation}
\text{ the distributional transform of $X$ is the random variable } T_X(X,U). 
\end{equation}
For a random vector the distributional transform is just the vector of the distributional transforms of its elements, e.g.,
\begin{equation}
T_{X_i}(X_i,U_i):= (T_{X_{i,1}}(X_{i,1},U_{i,1}),\ldots,T_{X_{i,d_i}}(X_{i,d_i},U_{i,d_i}))
\end{equation}
with $U_i = (U_{i,1},\ldots,U_{i,d_i})$ being a vector of independent uniformly distributed random variables. Note, if $X$ has a continuous distribution the distributional transform becomes just the classical transformation using the distribution function, i.e., $T_X(X,U) = F_X(X)$.
 
Based on \cite[Theorem 2.1]{Boet2020a} we know (for univariate $X$ and independent uniformly distributed $U$): $ T_X(X,U)$ is uniformly distributed, $T_{g(X)}(g(X),U) = T_X(X,U)$ for all strictly increasing functions $g$. Moreover, the key property for our context is the fact that random vectors $X_1,\ldots,X_n$ are independent if and only if $T_{X_1}(X_1,U_1),\ldots,T_{X_n}(X_n,U_n)$ are independent (where each $U_i$ is a vector of independent uniformly distributed random variables with the same dimension as $X_i$). Moreover, one can show to following identity.
\begin{proposition} \label{thm:dist-trans-properties} Let $X$ and $U$ be independent univariate random variables, $U$ be uniformly distributed and $g$ be a strictly decreasing function on the range of $X.$ Then
\begin{equation}
T_{g(X)}(g(X),U) = 1-T_X(X,1-U). 
\end{equation} 
\end{proposition}
\begin{proof}
Let $g$ be strictly decreasing, $x\in\R$, $u\in [0,1],$ then
\begin{align*}
T_{g(X)}(g(x),u) &= \Prob(g(X) < g(x)) + u\Prob(g(X) = g(x))\\
& = \Prob(X>x) + u \Prob(X= x)\\
&= 1 - (\Prob(X< x) + (1-u) \Prob(X=x))\\
&= 1- T_X(x,1-u). \qedhere
\end{align*} 
\end{proof}

For a dependence measure $d$ define
\begin{equation}
d_{cop}(X_1,\ldots,X_n):=d(T_{X_1}(X_1,U_1),\ldots, T_{X_n}(X_n,U_n)).
\end{equation}

\begin{theorem}\label{thm:dcop-invariance}
Let $d$ be translation invariant and symmetric in each 'element',  then $d_{cop}$ is invariant with respect to all strictly monotone transformations $g_{i,k}$, $i=1,\ldots,n$, $k= 1,\ldots,d_i$, i.e., 
\begin{equation}\label{eq:dcop-invariance}
d_{cop}(X_1,\ldots,X_n) = d_{cop} (g_1(X_1),\ldots,g_n(X_n))
\end{equation} 
where $g_i(X_i):= (g_{i,1}(X_{i,1}),\ldots,g_{i,d_i}(X_{i,d_i}))$ for each $i$.
\end{theorem}
\begin{proof}
For each element remove $g_{i,k}$ by the identities $T_{g(X)}(g(X),U) = T_X(X,U)$ and $T_{g(X)}(g(X),U) = 1-T_X(X,1-U).$ Next use the translation invariance and symmetry of $d$ to transform any $1-T_{X_{i,k}}(X_{i,k},1-U_{i,k})$ into  $T_{X_{i,k}}(X_{i,k},1-U_{i,k})$. Altogether yielding elements of the form $T_{X_{i,k}}(X_{i,k},\widetilde U_{i,k})$ with $\widetilde U_{i,k} = U_{i,k}$ if $g_{i,k}$ was increasing and $\widetilde U_{i,k} = 1-U_{i,k}$ if $g_{i,k}$ was decreasing. Finally, note that due to the independence of the uniformly distributed $U_{i,k}$ the joint distribution of $\widetilde U_{i,k}$ is identical to that of the $U_{i,k}$ and hence also the joint distribution of $T_{X_{i,k}}(X_{i,k},\widetilde U_{i,k}),k=1,\ldots,n$ is equal to the joint distribution of $T_{X_{i,k}}(X_{i,k}, U_{i,k}),k=1,\ldots,n.$ 
\end{proof}
 
Note, by Remark \ref{rem:distribution-free} a measure $d_{cop}$ with property \eqref{eq:dcop-invariance} is (for univariate continuous marginals) distribution-free. Moreover, it is also element-wise symmetric since $g_{i,k}$ in Theorem \ref{thm:dcop-invariance} can be increasing or decreasing.
 
The sample version of the distributional transform (for samples $x$ and $x^{(1)},\ldots,x^{(N)}$ of $X$ (univariate) and $u$ a sample of $U$) is given by
\begin{equation}
\hN T(x,u; x^{(1)},\dots, x^{(N)}) = \frac{1}{N} \sum_{k =1}^N \left[\bm{1}_{(-\infty, x)}(x^{(k)}) + u \bm{1}_{\{x\}}(x^{(k)})\right].
\end{equation}
In a multivariate setting each element of a vector is transformed in the same way and for each element $u$ is independently sampled from a uniform distribution, hence this is also called \emph{Monte Carlo distributional transform} (for further details see \cite[Section 2]{Boet2020a}). Figure \ref{fig:dist-trans-letter} illustrates the results of the Monte Carlo distributional transform for the data of Figure \ref{fig:letters}.

\subsection{The copula version of distance multicorrelation}
\label{sec:copversions} 

Based on the above the copula version of total distance multicorrelation is defined by
\begin{equation}
C\overline{Mcor}(X_1,\ldots,X_n):= \overline{Mcor}(T_{X_1}(X_1,U_1),\ldots, T_{X_n}(X_n,U_n)).
\end{equation}
The empirical version of it is just the known estimator of the measure applied to the Monte Carlo distributional transform of the samples. For this joint estimators the convergence was discussed in \cite{Boet2020a}.

Of the examples in the previous sections Example \ref{ex:changeofdependence} and Figure \ref{fig:bias-2} apply directly also to the copula version, since therein all univariate marginals are uniformly distributed. Hence, the value of the measure depends on the dependence within a multivariate margin and its bias corrected estimator has a large variance for small samples.

Nevertheless, the key reason to introduce the copula versions was the systematic influence of marginal distributions on values of total distance multicorrelation which was discussed in Example \ref{ex:changeofdistribution}. Now Figures \ref{fig:pair-freq-cop} and \ref{fig:pair-freq-cop-full} illustrate that the copula version overcomes this shortcoming.

\begin{figure}[H]
%\newlength{\myl}
\setlength{\myl}{0.47\textwidth}
\centering
\includegraphics[width = \myl]{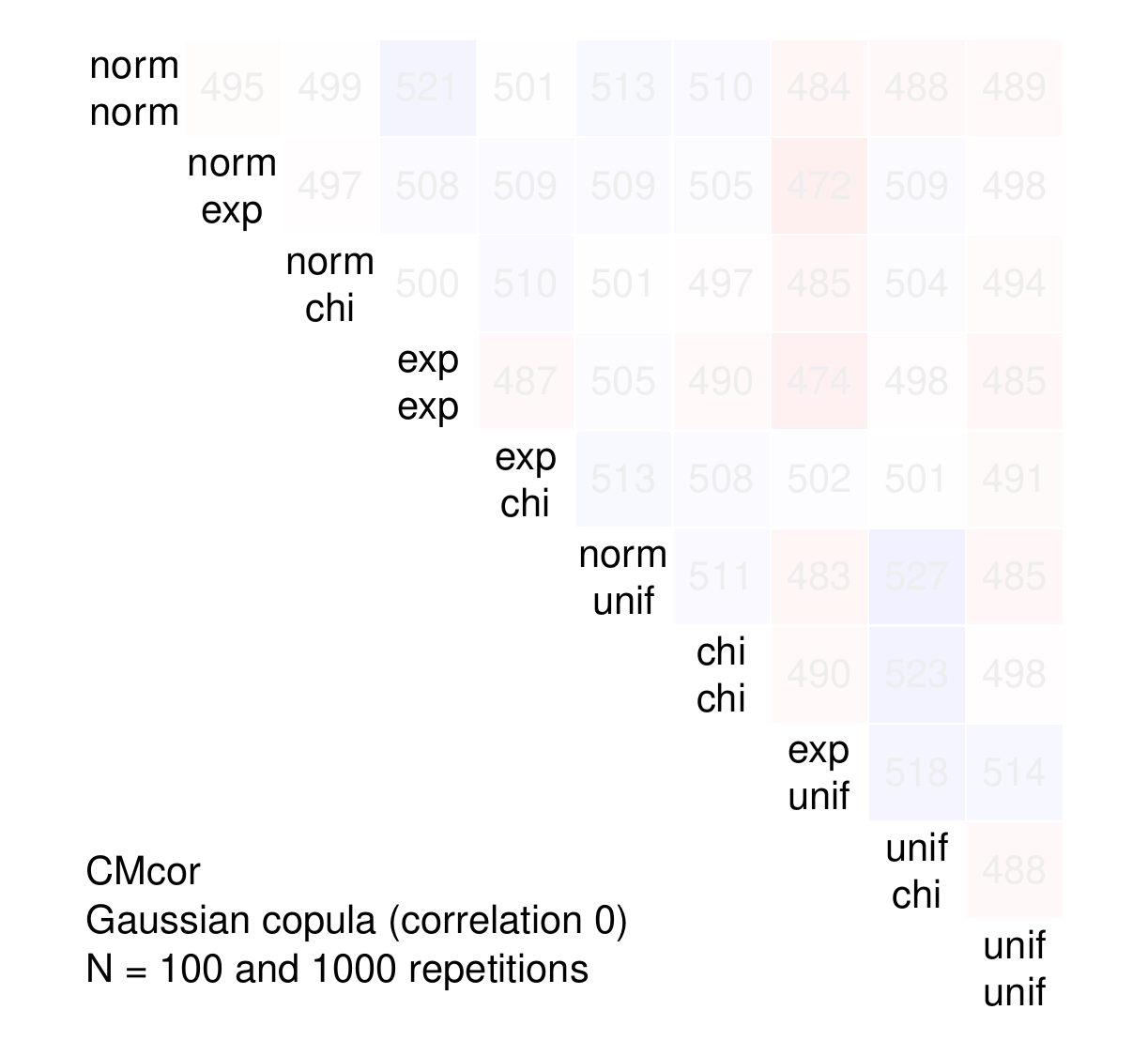}
\includegraphics[width = \myl]{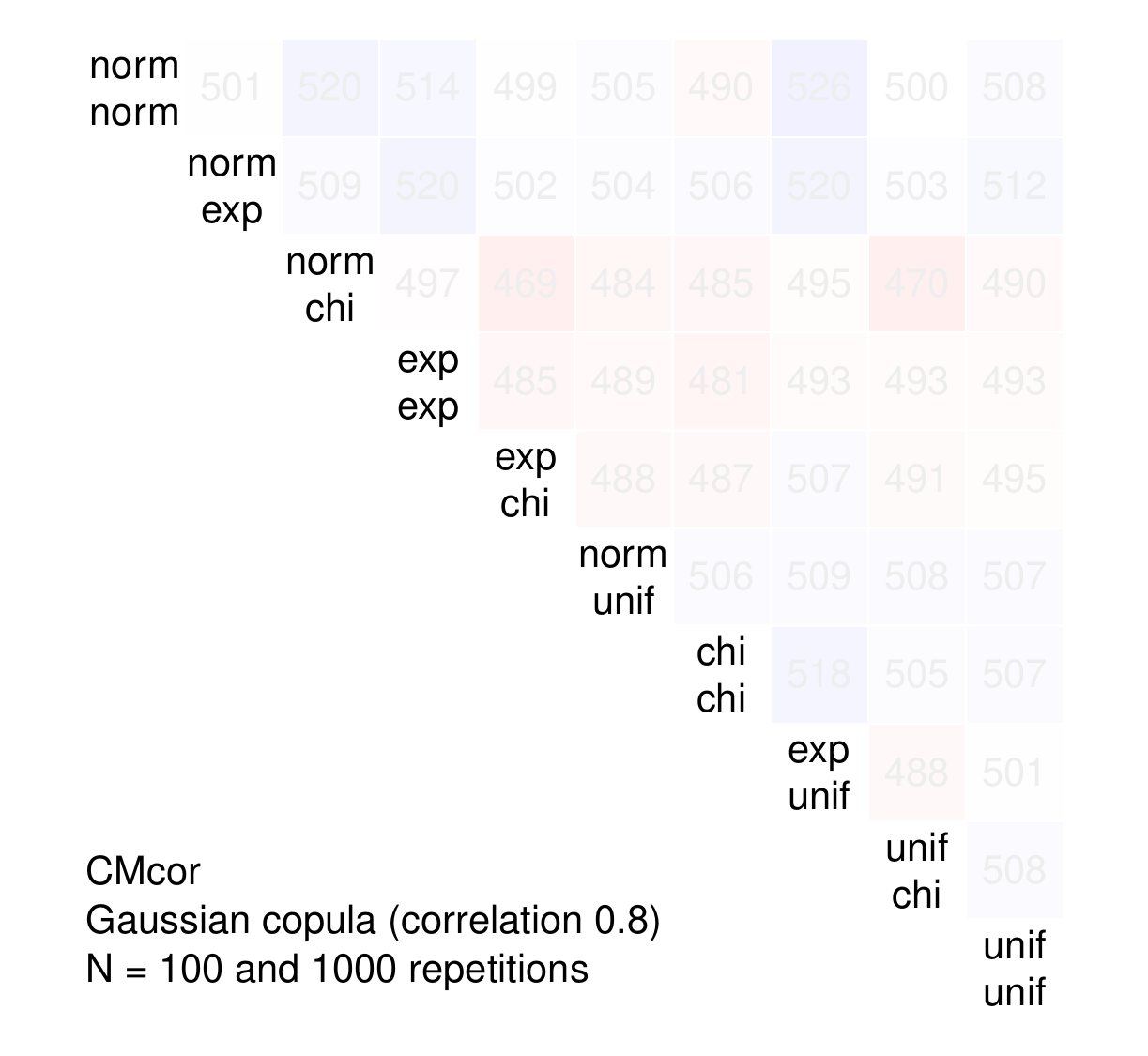}

\caption{Invariance with respect to marginal distributions. Same as Figures \ref{fig:pair-freq} and \ref{fig:pair-freq-unbiased} but using the copula version of distance multicorrelation. Observation: There is no systematic dependence on the marginal distributions for the copula version of distance multicorrelation.}		\label{fig:pair-freq-cop}
\end{figure}

The copula version of distance multicorrelation does not require that the initial marginals have a continuous distribution. For discontinuous distributions it fixes (compare with Remark \ref{rem:distribution-free}) the underlying copula to be the multilinear extension copula, see \cite{Boet2020a} for details. This is in some sense the choice which introduces as much independence as possible while preserving an underlying dependence in the sample. In fact any other choice would fail to preserve independence. As a consequence of this the value obtained for the dependence of discrete (especially Bernoulli) samples is not always appropriate for a comparison with cases with continuous marginals. In these cases it might be understood as a lower bound. A possible upper bound could be obtained by using the same uniform sample in the distributional transform for each element, but in general this upper-bound would fail to characterize independence.

\begin{example}[Comparabilty of discrete distributions] In the setting of Example \ref{ex:changeofdistribution} one gets 
\begin{align} 
\vapprox{CMcor(X_{exp},Y_{chi})}{0.74} &= \vapprox{CMcor(X_{norm},Y_{chi})}{0.74} = \vapprox{CMcor(X_{unif},Y_{chi})}{0.74}\\
&> \vapprox{CMcor(X_{exp},Y_{bern})}{0.57} = \vapprox{CMcor(X_{chi},Y_{bern})}{0.57}.
\end{align}
Which illustrates that the distributional transform yields smaller values for discrete marginals. For a multivariate example see Figure \ref{fig:multivariate-copula}. Thus also the copula version of distance multicorrelation can cause systematic errors in general. But the class for which such errors can occur is much smaller than in the case without the distributional transform.
\end{example}

 \begin{figure}[H]
 %\newlength{\myl}
 \setlength{\myl}{0.99\textwidth}
 \centering
 \includegraphics[width = \myl]{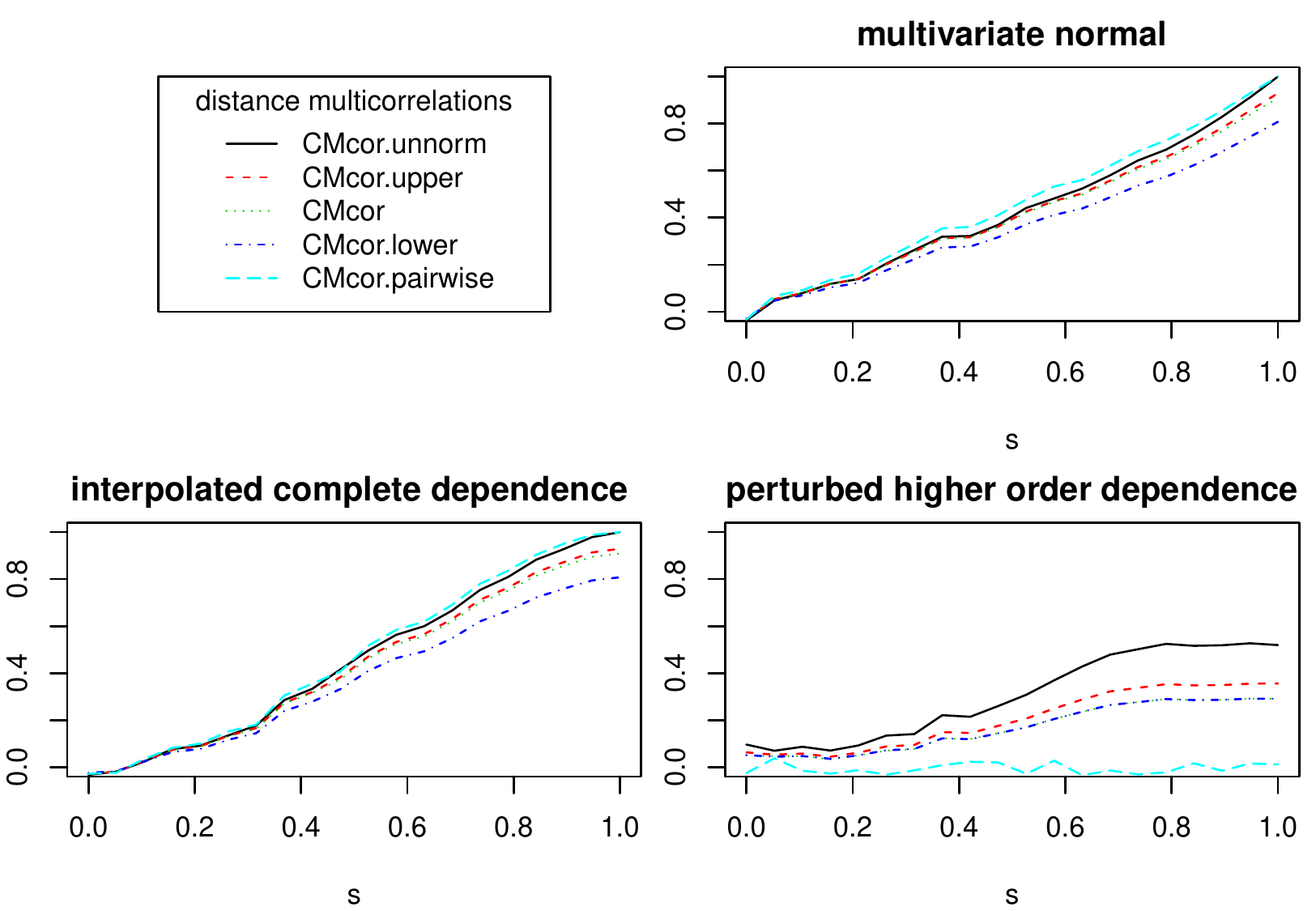}
 
 \caption{Illustrations of the values of copula distance multicorrelations for Example \ref{ex:multivariate}. Observation: The values are similar to those without the distributional transform (Figure \ref{fig:multivariate}). In the case of higher order dependent Bernoulli variables here all measures are bounded by 1 an their values are less than those in Figure \ref{fig:multivariate}.}		\label{fig:multivariate-copula}
 \end{figure}

\section{Conclusion} \label{sec:conclusion}

Several aspects have been discussed in the previous sections. Some key observations are:

\begin{itemize}
\item 
{\bf Correlation} can not characterize independence. It measures in some sense linear dependence, and hence it is (only) in a linear regression approach an appropriate measure. In general, values besides -1 and 1 only provide a somewhat arbitrary scale / order of dependence. Moreover its empirical estimator is (as it is the case also for all other discussed measures) biased.
\item 
{\bf Distance multicorrelation} (including {\bf distance correlation}) characterizes independence. In a regression setting its values turn out to behave similar to those of correlation, but they can also be used in general to quantify dependence. Here at least 0 (and 1 in certain settings) have clear interpretations. Nevertheless, for multivariate marginals the value depends on the dependence within the marginals, thus it does not quantify exclusively the dependence of the random vectors (this seems to apply in fact to any non-trivial multivariate dependence measure). Moreover, if the marginal distributions are varied also the values of the measure vary. Hence a direct comparison of values can yield in certain settings severe systematic errors.
\item
{\bf Copula distance multicorrelation} characterizes independence, and it is invariant with respect to changes (of continuous) marginal distributions. For univariate continuously distributed random variables it is an appropriate measure of dependence. But if the random variables under consideration are multivariate one has to keep in mind, that its value depends also on the internal dependence of the random vectors. Moreover, dependencies of non-continuous marginals are underestimated in comparison to continuous marginals.
\item
In contrast to the above measures, the corresponding {\bf p-values of the independence tests} have in any setting a clear interpretation: the likelihood of the given sample (or worse) in the case of independence of the components.
\end{itemize}

We hope that the presented discussions and examples help to understand the limitations and use of dependence measures in general and distance multicorrelation in particular. Let us close with a general appeal:
\medskip

\textit{If you teach statistics, mention at least some proper dependence measure. If you analyze dependence, use some proper dependence measure and know its limitations. If you quantify strength of dependence, be aware that only certain values have a clear interpretation -- in general the obtained values will just be numbers on a scale, which might not exclusively describe dependence. In contrast, the p-values of the corresponding independence tests contain at least well understood information on the likelihood.}

\section{Figures}
\newlength{\mylength}
\setlength{\mylength}{0.80\textwidth}

\begin{figure}[H]
\centering
\includegraphics[width = \textwidth]{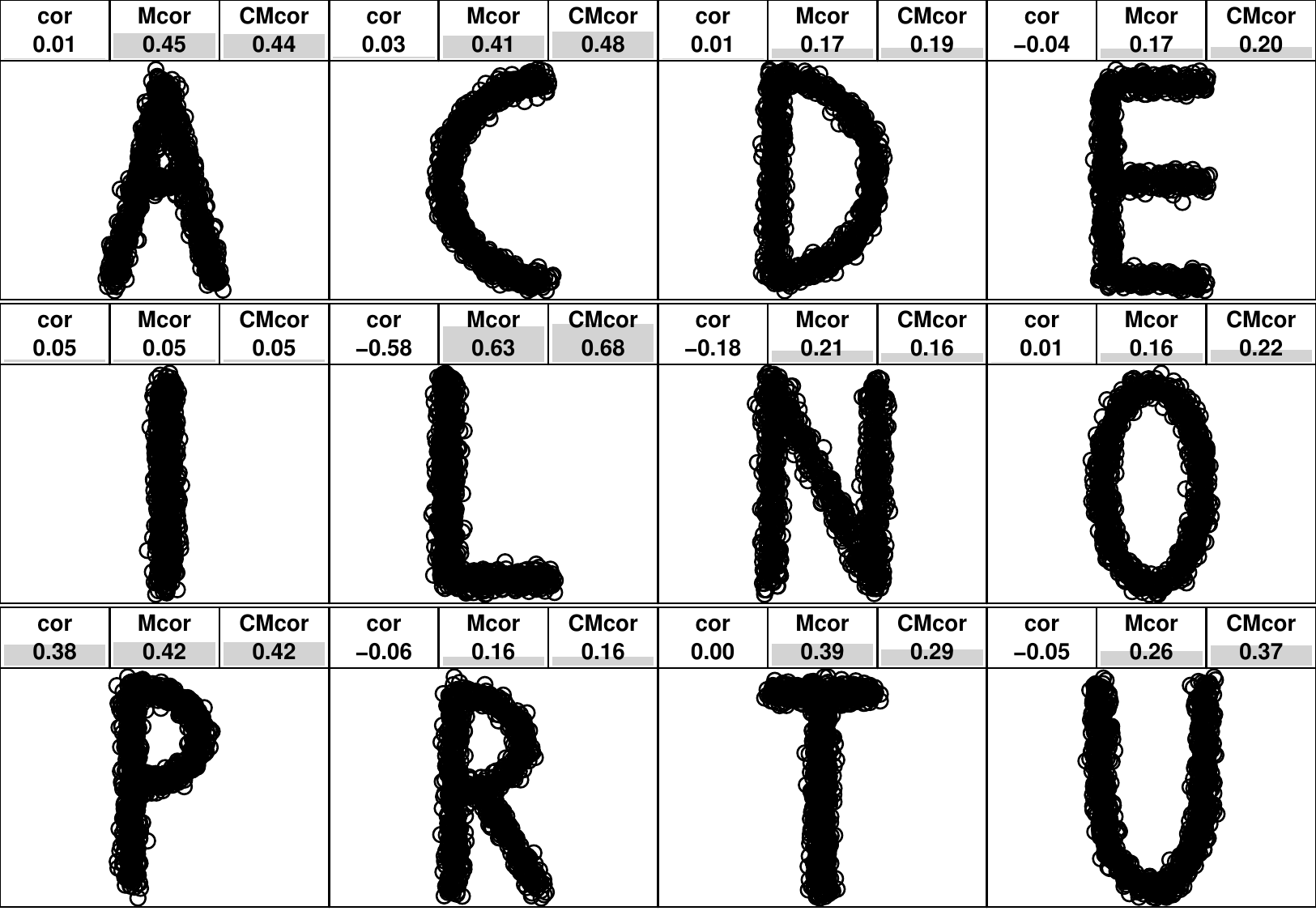}
\caption{The letters of UNCORRELATED and INDEPENDENT. Based on 1000 samples each using a uniform distribution on a line representation of the letter perturbed by a bivariate normal distribution (with independent components). Observation: In most cases Pearson's correlation fails to detect the dependencies, but they are detected by distance multicorrelation and its copula version. Theoretically only for 'I' the variables are really independent.}		\label{fig:letters}
\end{figure} 
\vfill
\newpage

\begin{figure}[H]
\centering
\includegraphics[width =\mylength]{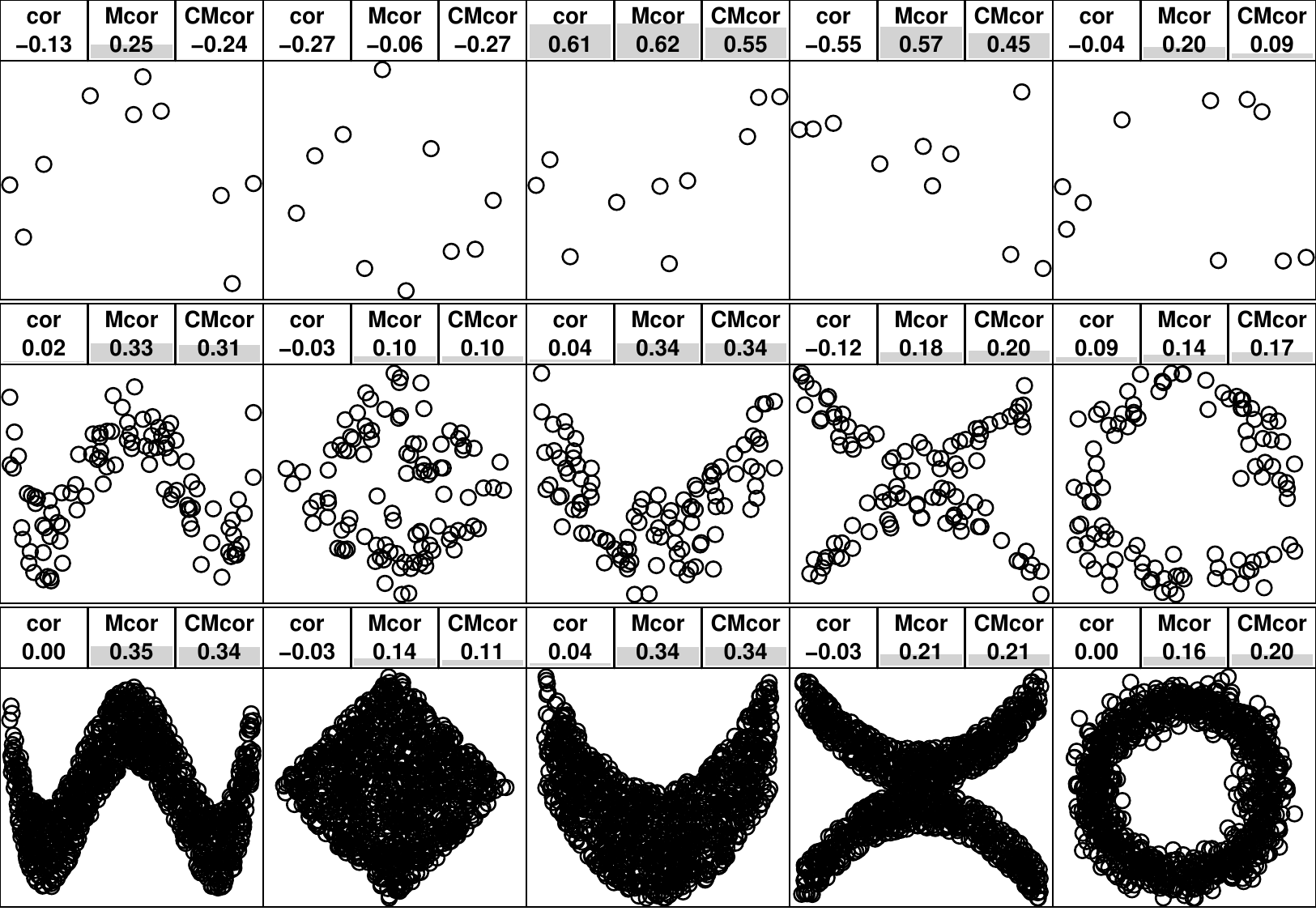}
\caption{Classical examples with 0 correlation but dependence. The rows are based on 10, 100 and 1000 samples, respectively. Observation: Pearson's correlation fails to detect the dependencies, but they are detected by distance multicorrelation and its copula version.}		\label{fig:classical-cor-dcor}
\end{figure} 

\begin{figure}[H]
\centering
\includegraphics[width =\mylength]{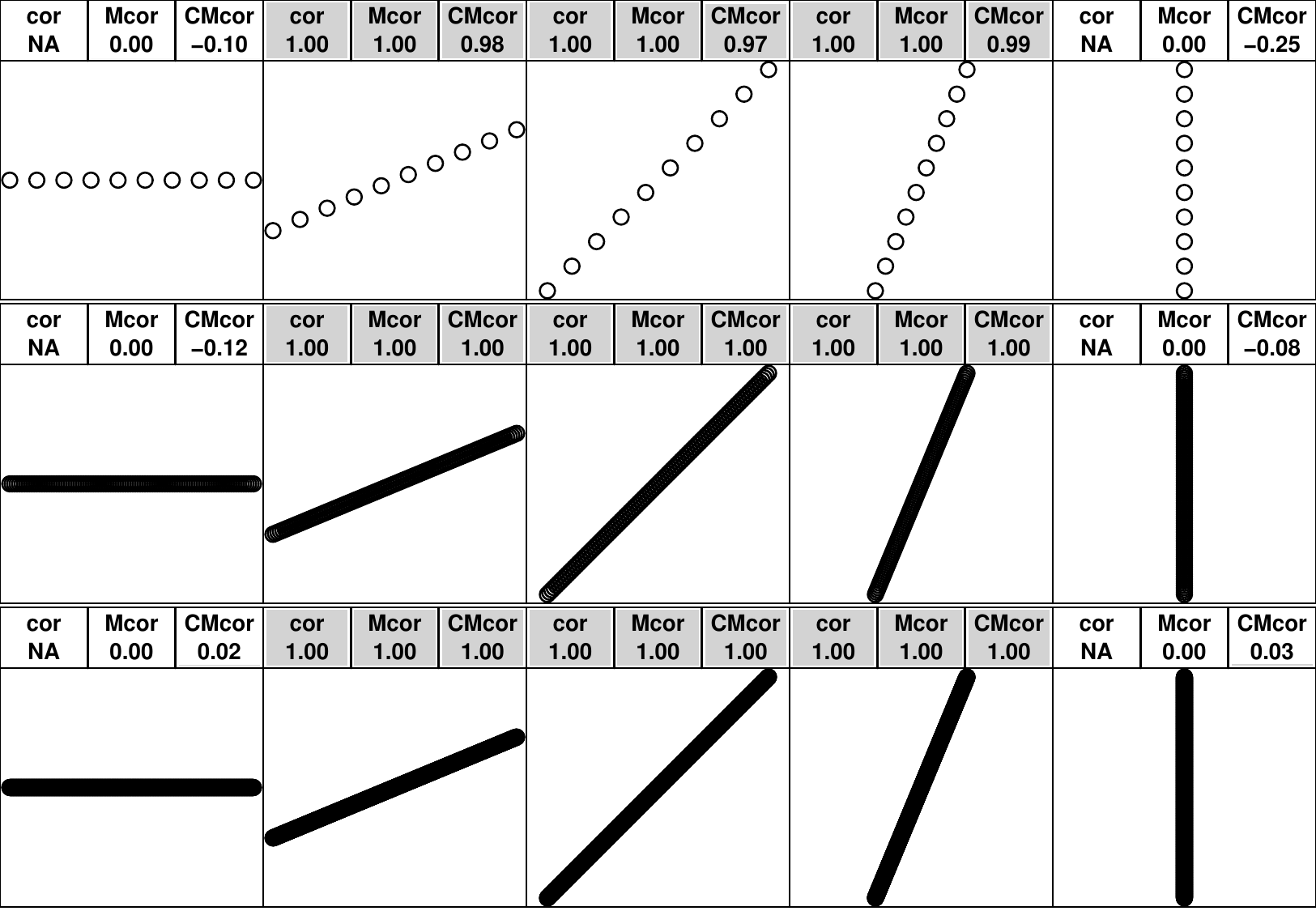}
\caption{Linearly related variables - a reminder that correlation does not measure linear relation but linear dependence. The correlation is 0 (or undefined) for the first and last column. Hence, constant variables should be removed before any analysis, since they are independent anyway. The rows are based on 10, 100 and 1000 samples, respectively.}		\label{fig:linear-relation}
\end{figure}

\begin{figure}[H]
\centering
\includegraphics[width = 0.95\textwidth]{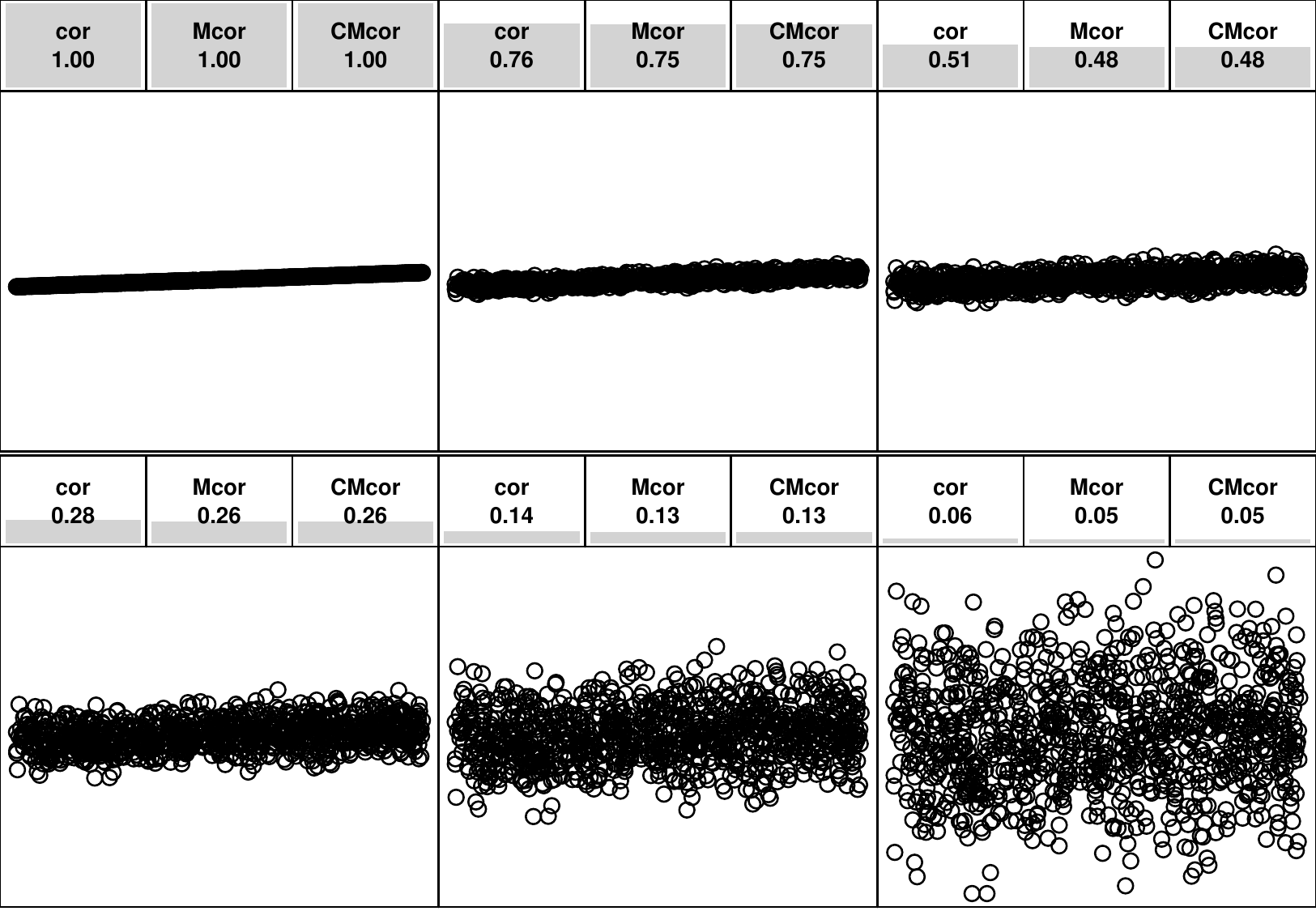}
\includegraphics[width = 0.95\textwidth]{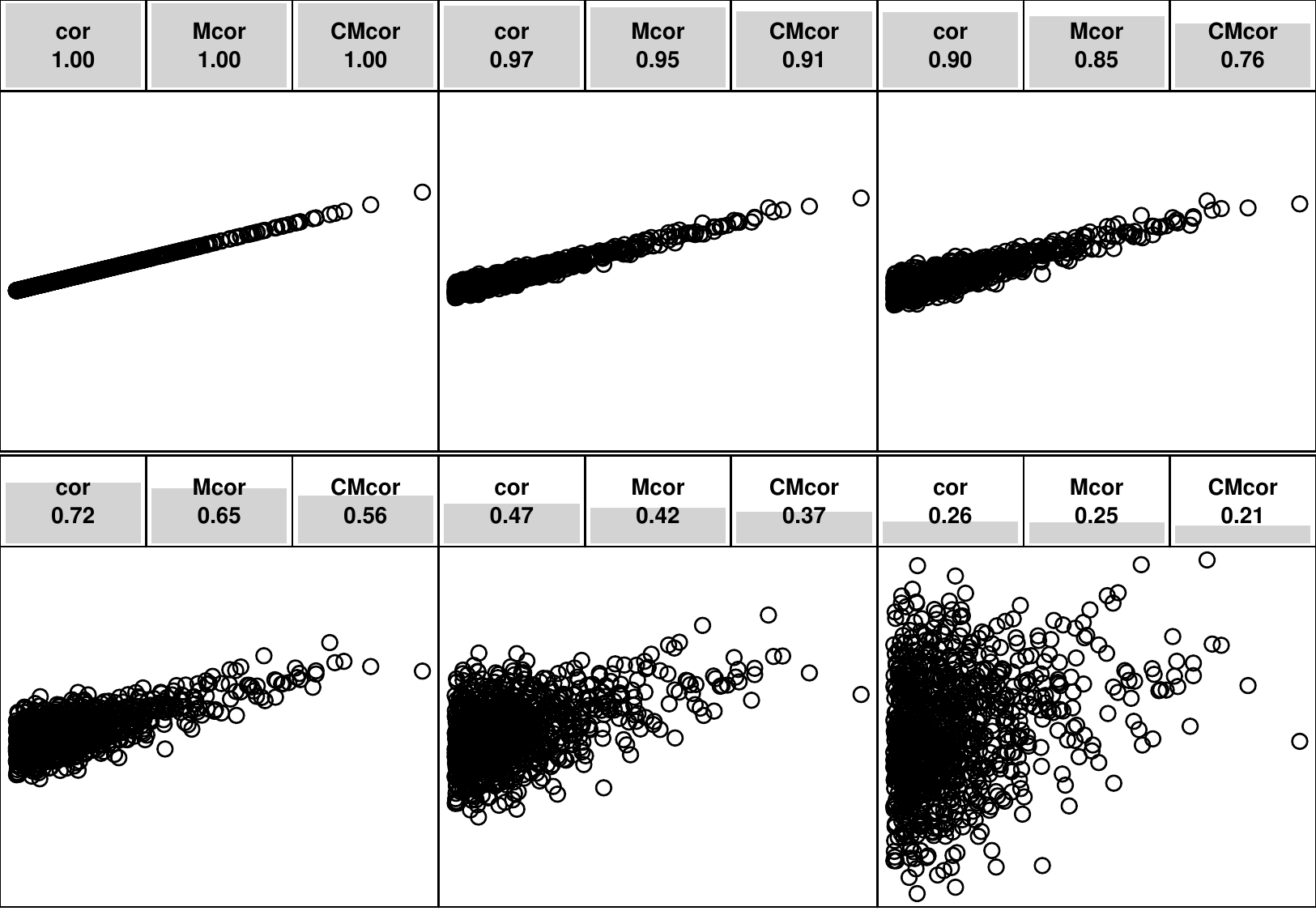}
\caption{Samples of $(X,X+rZ)$ where $X$ is uniformly distributed (first two rows) or exponentially distributed (last two rows) and $Z$ is standard normally distributed with $r = 0, 0.25, 0.5, 1, 2$ and $4$, respectively. Observation: In a regression setting Pearson's correlation seems more sensitive than the other measures. But qualitatively all behave similarly.}		\label{fig:linear-relation-perturbed}
\end{figure}

\begin{figure}[H]
\includegraphics[width = \textwidth]{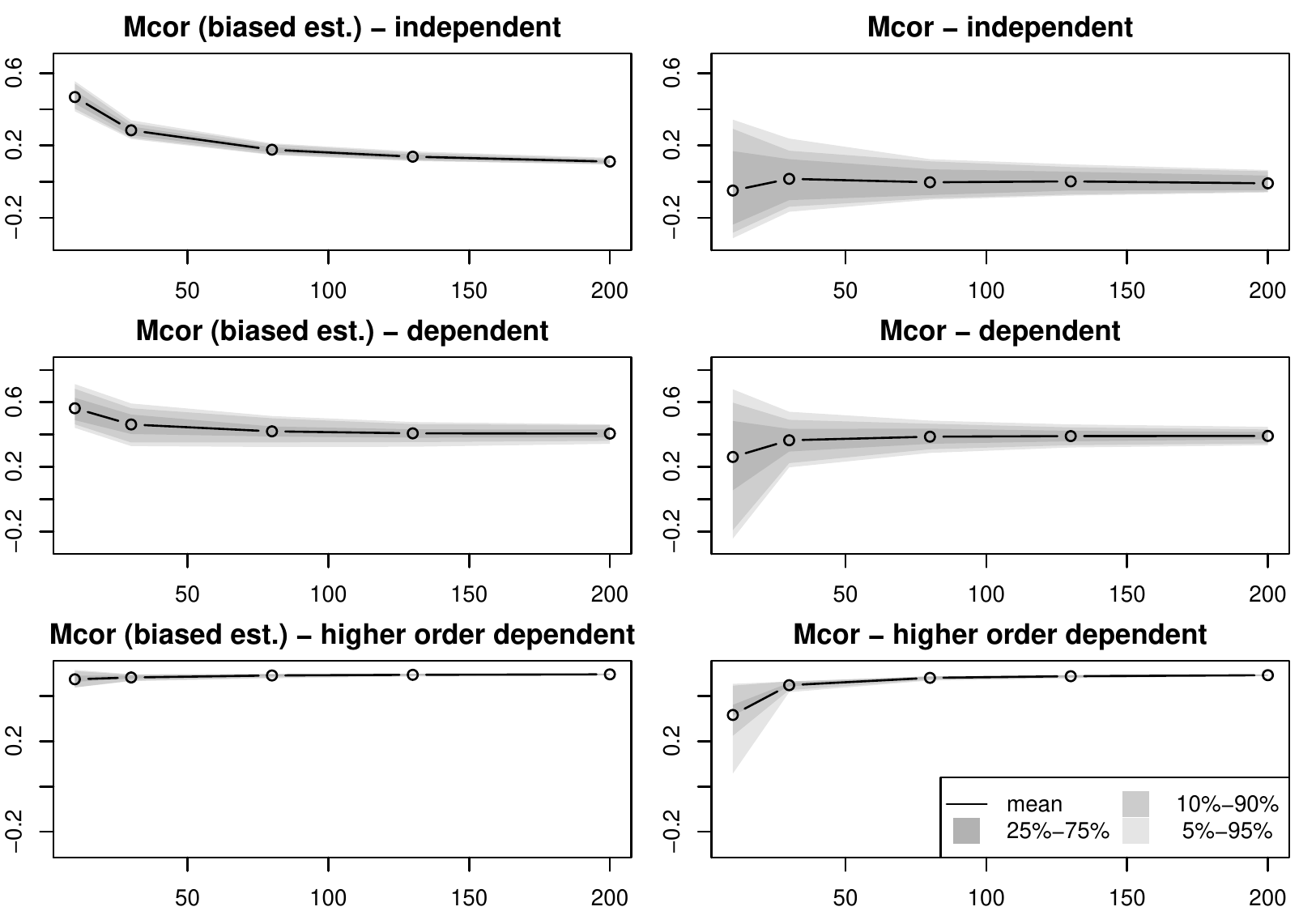}
\caption{Comparison of biased and bias corrected estimates for $n=3$. The first two rows are the examples of Figure \ref{fig:bias-2} for $n=3$, the third row is an example of three Bernoulli random variables where the first two are independent and the third is 1 if and only if the two others have the same value. Then the three variables are dependent, but pairwise independent. Observation: For the first two examples it is clearly visible that the bias corrected estimate have less bias, but they feature a larger variance. For the higher order dependence the variance difference is less dramatic.}		\label{fig:bias-multi-1}
\end{figure}

\begin{figure}[H]
\includegraphics[width = \textwidth]{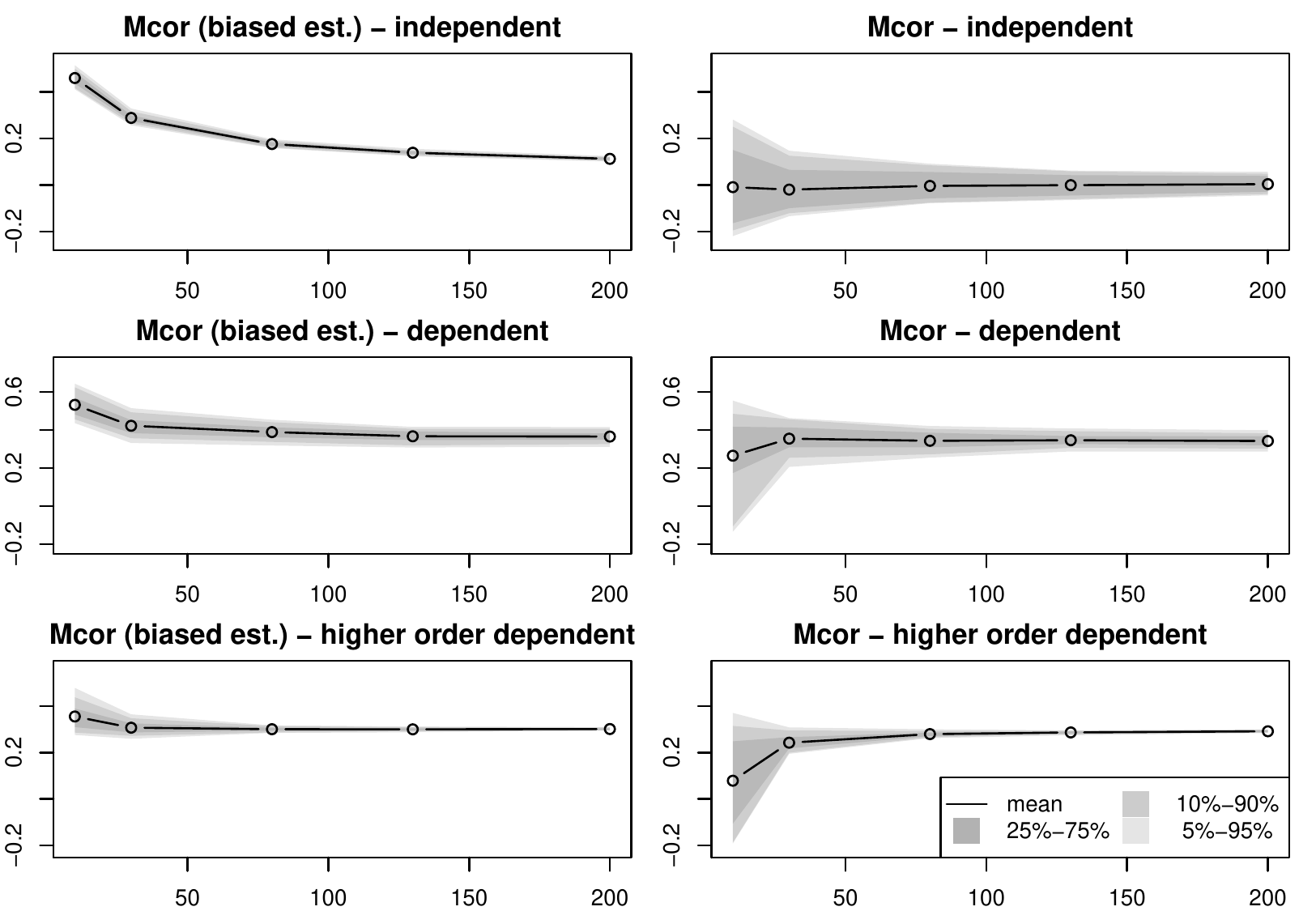}
\caption{Comparison of biased and bias corrected estimates for $n=4$. The details are the same as in Figure \ref{fig:bias-multi-1}.}		\label{fig:bias-multi-2}
\end{figure}

\begin{figure}[H]
\includegraphics[width = 0.99\textwidth]{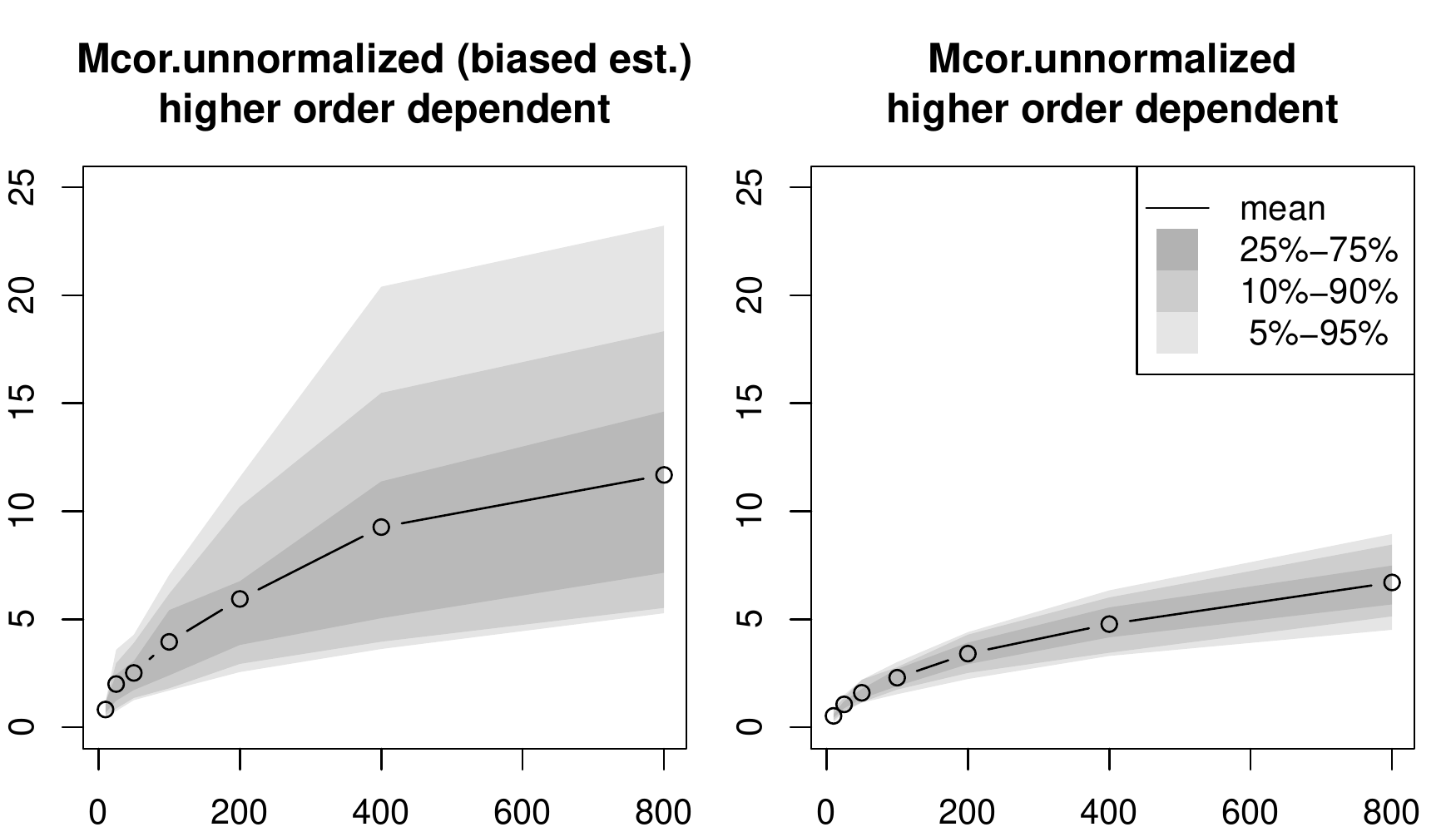}
\caption{Unnormalized total distance multicorrelation. The estimators for unnormalized distance multicorrelation show erratic behaviour for samples of the higher order dependent random variables described in Figure \ref{fig:bias-multi-1}.}		\label{fig:bias-multi-3}
\end{figure}

\begin{figure}[H]
%\newlength{\myl}
\centering
\setlength{\myl}{0.49\textwidth}
\includegraphics[width = \myl]{./Figs/pair-frequencies-1.pdf}
\includegraphics[width = \myl]{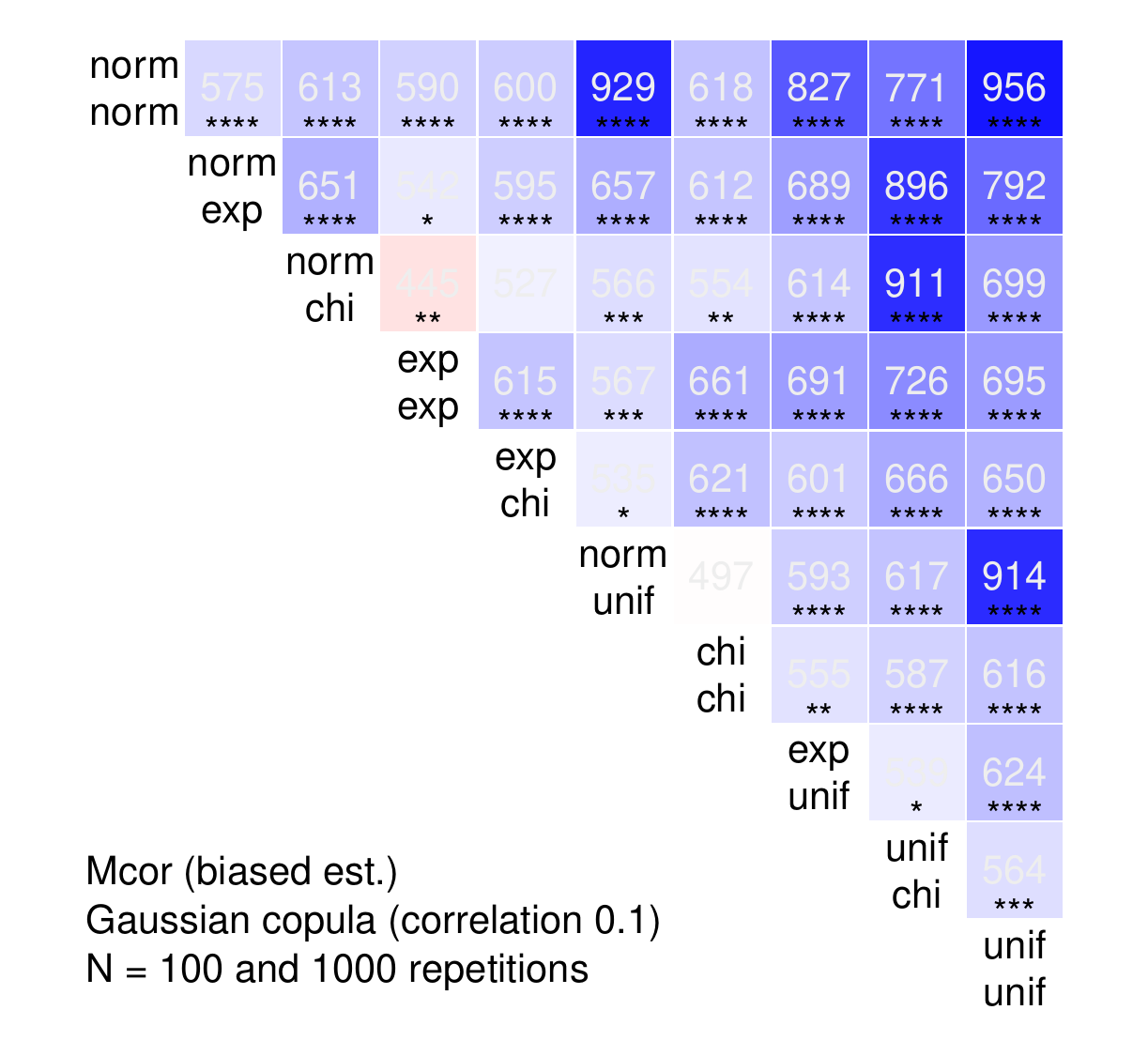}
\includegraphics[width = \myl]{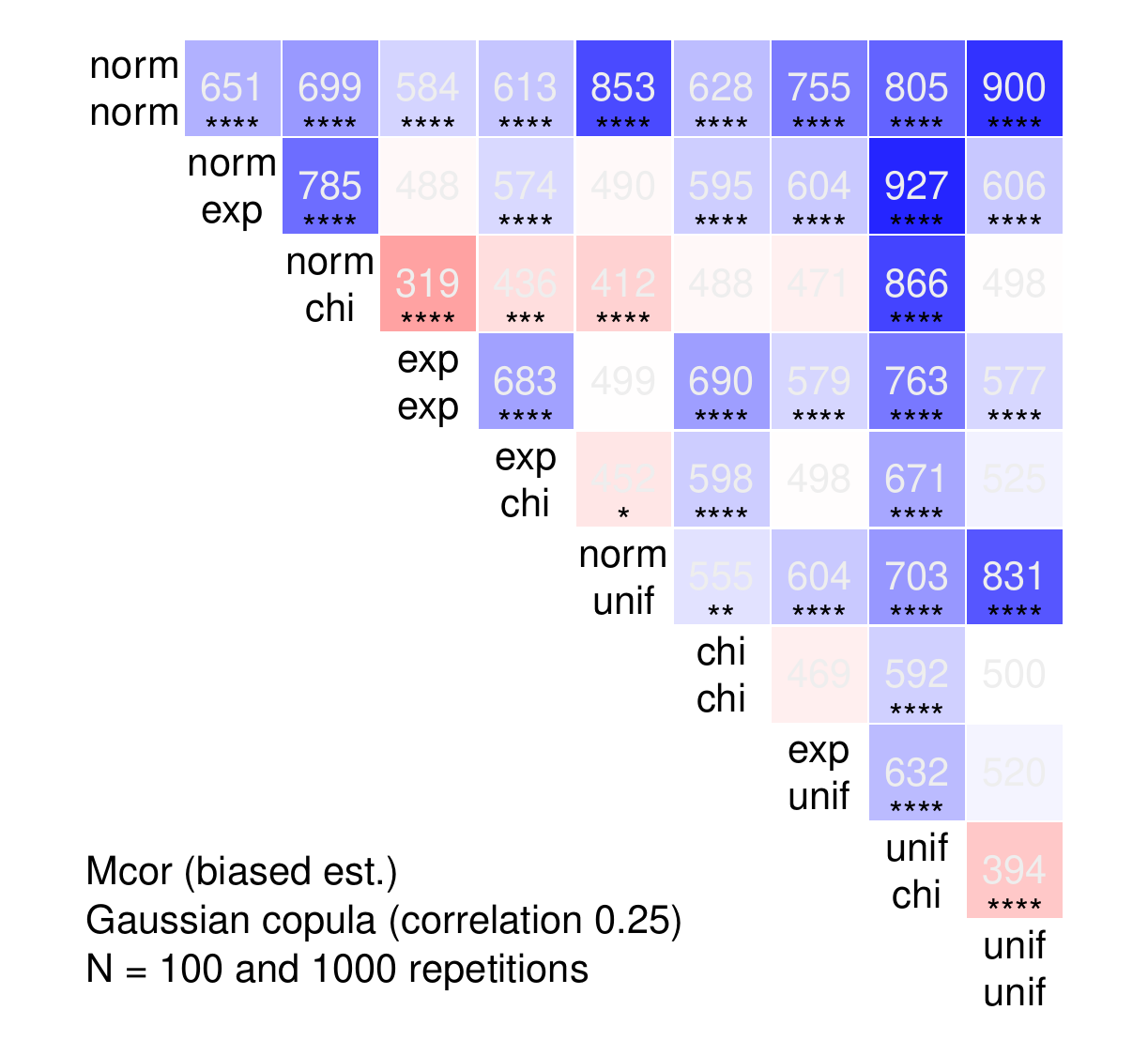} 
\includegraphics[width = \myl]{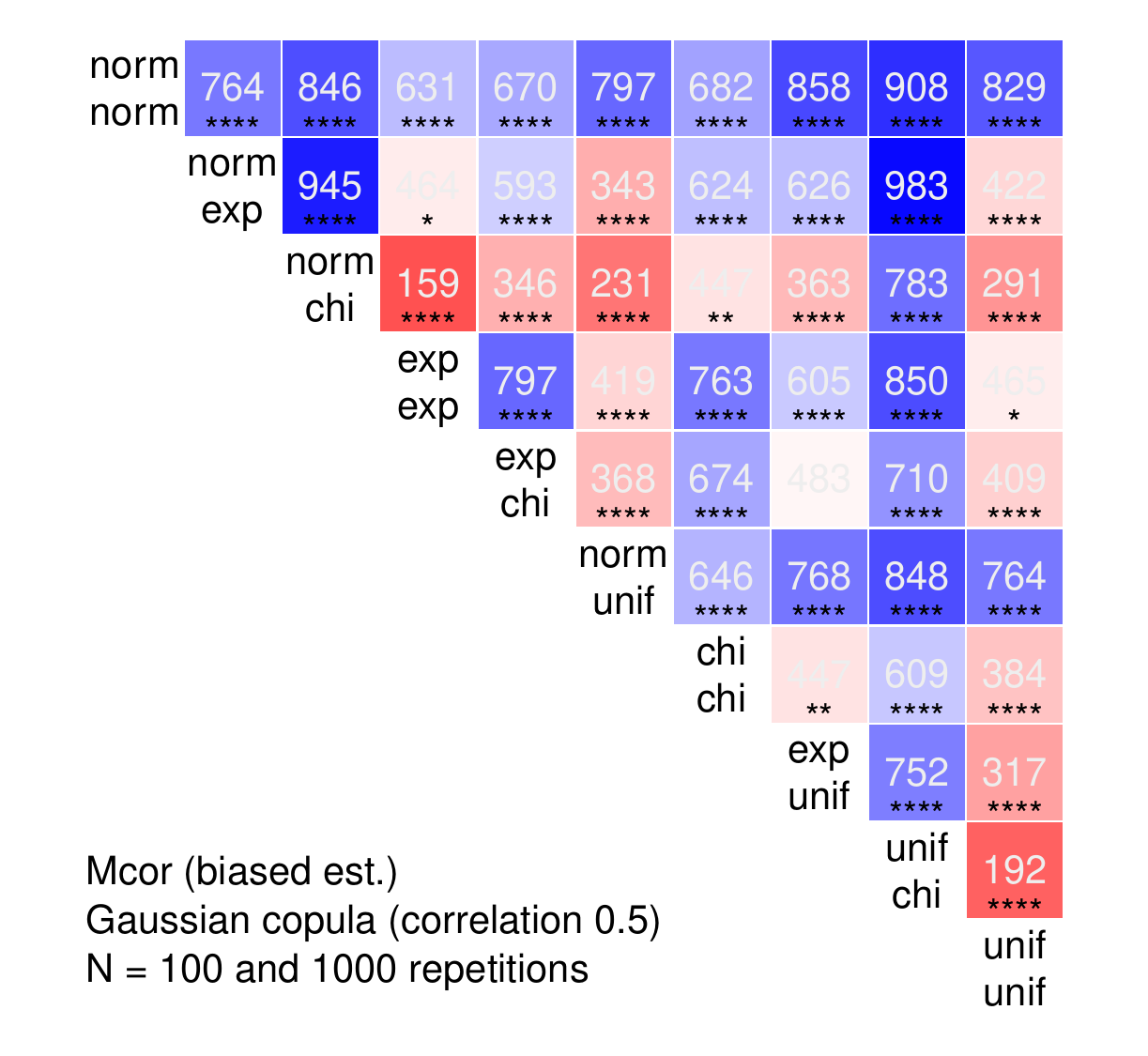}
\includegraphics[width = \myl]{./Figs/pair-frequencies-5.pdf}
\includegraphics[width = \myl]{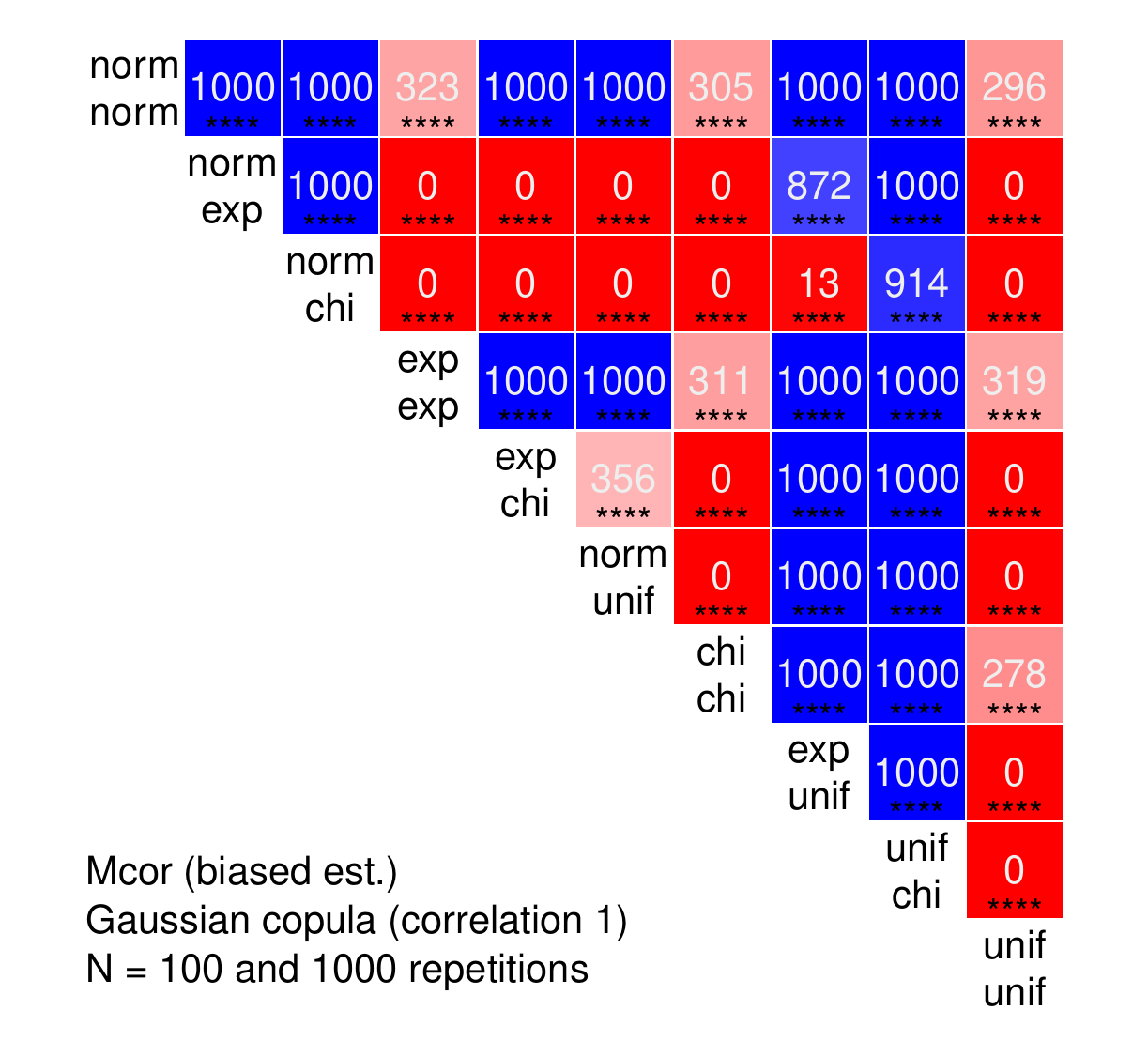} 
\caption{Systematic dominance due to marginal distributions. Extension of Figure \ref{fig:pair-freq} by considering 0, 0.1, 0.25, 0.5, 0.8 and 1 as parameters for the Gaussian copula.}		\label{fig:pair-freq-full}
\end{figure}

\begin{figure}[H]
%\newlength{\myl}
\centering
\setlength{\myl}{0.49\textwidth}
\includegraphics[width = \myl]{./Figs/pair-frequencies-umcor-1.pdf}
\includegraphics[width = \myl]{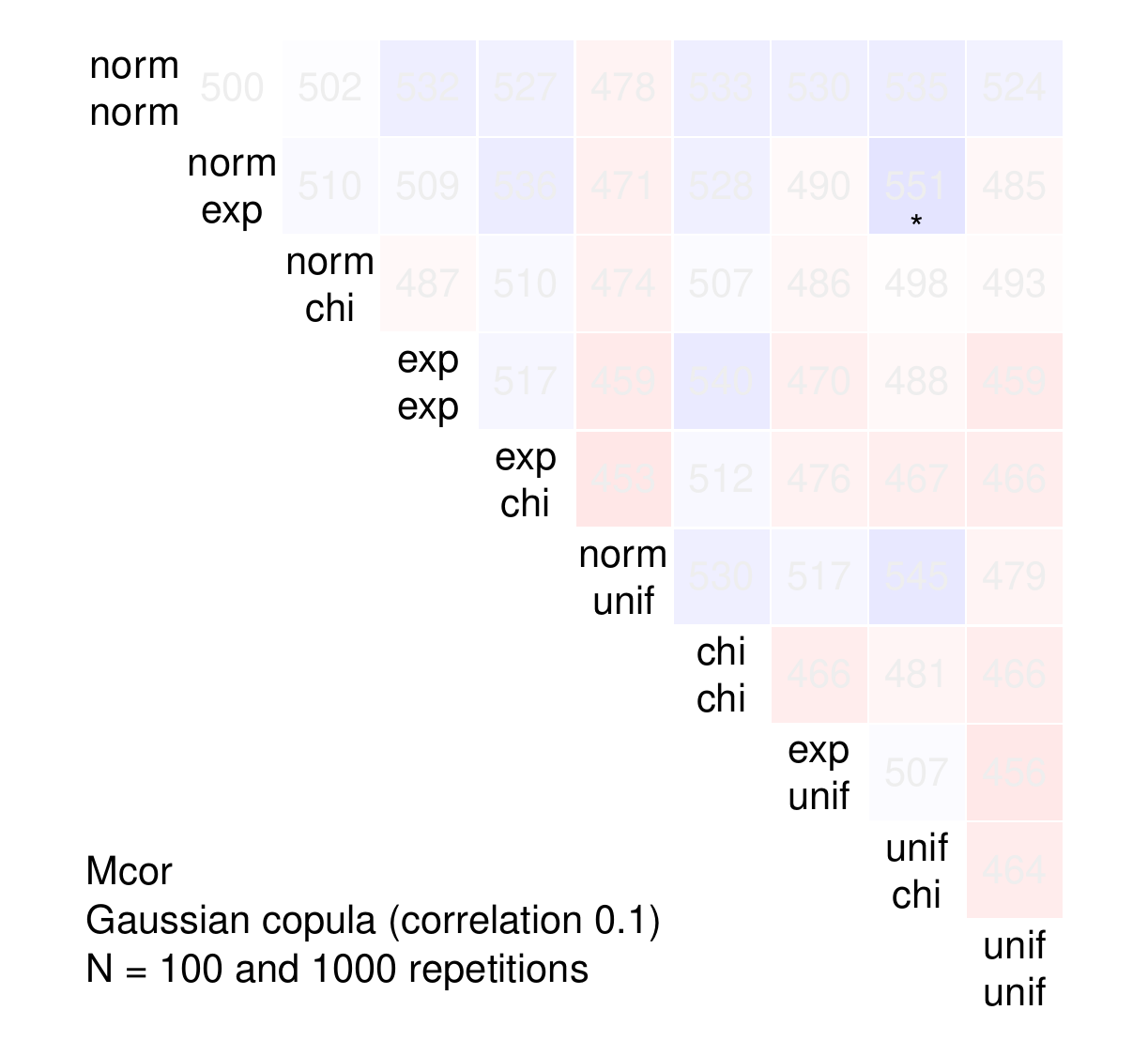}
\includegraphics[width = \myl]{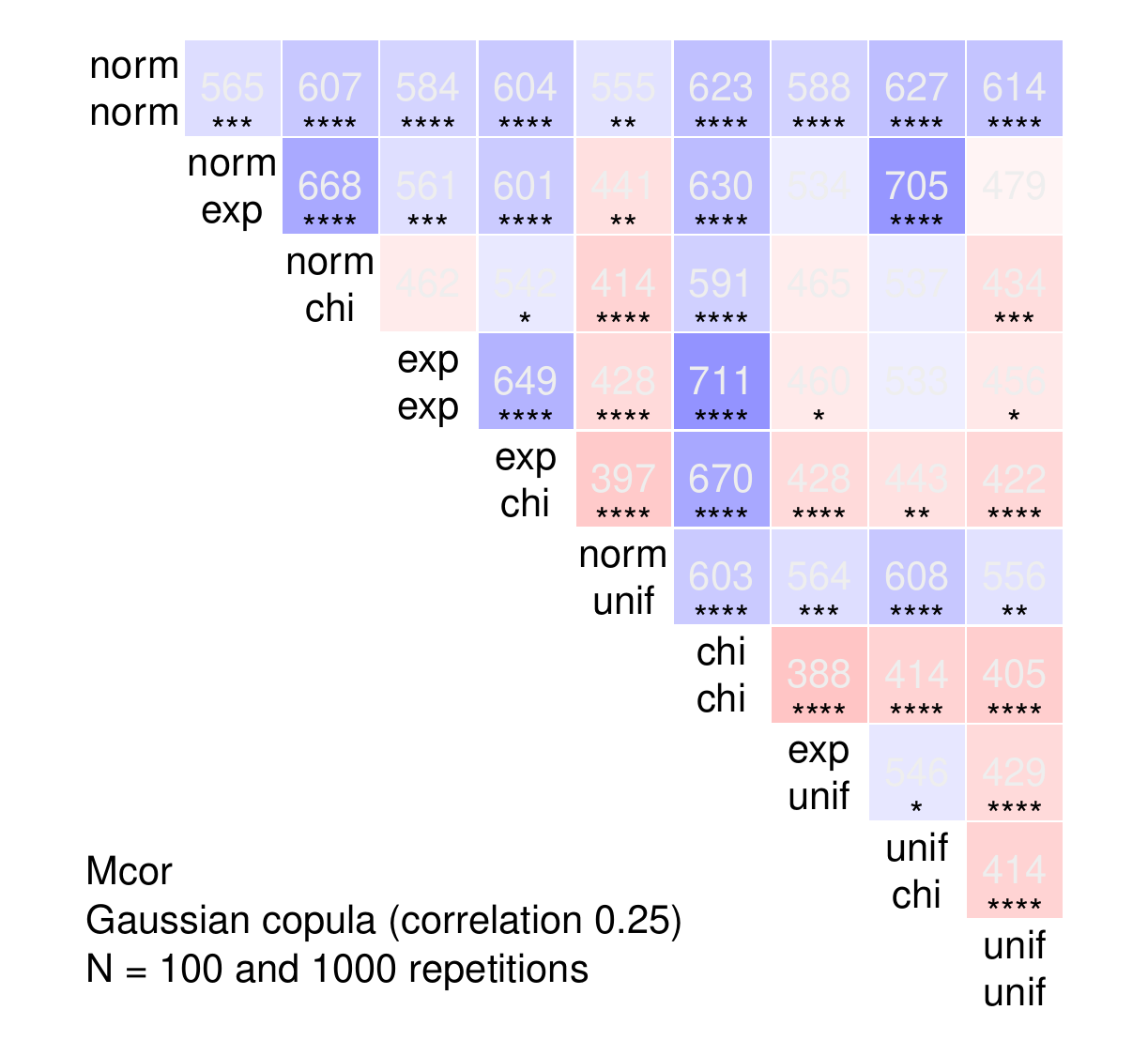} 
\includegraphics[width = \myl]{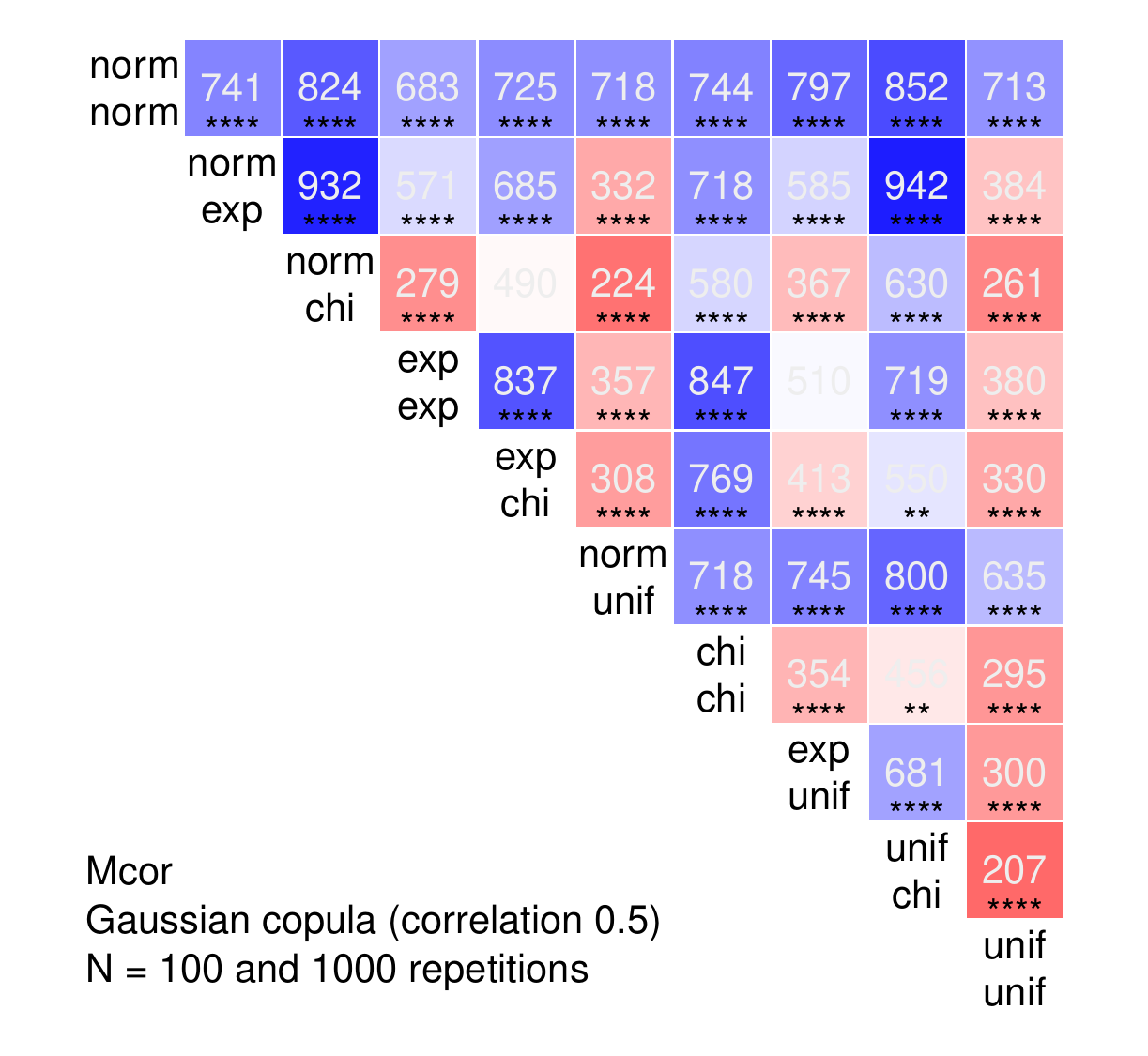}
\includegraphics[width = \myl]{./Figs/pair-frequencies-umcor-5.pdf}
\includegraphics[width = \myl]{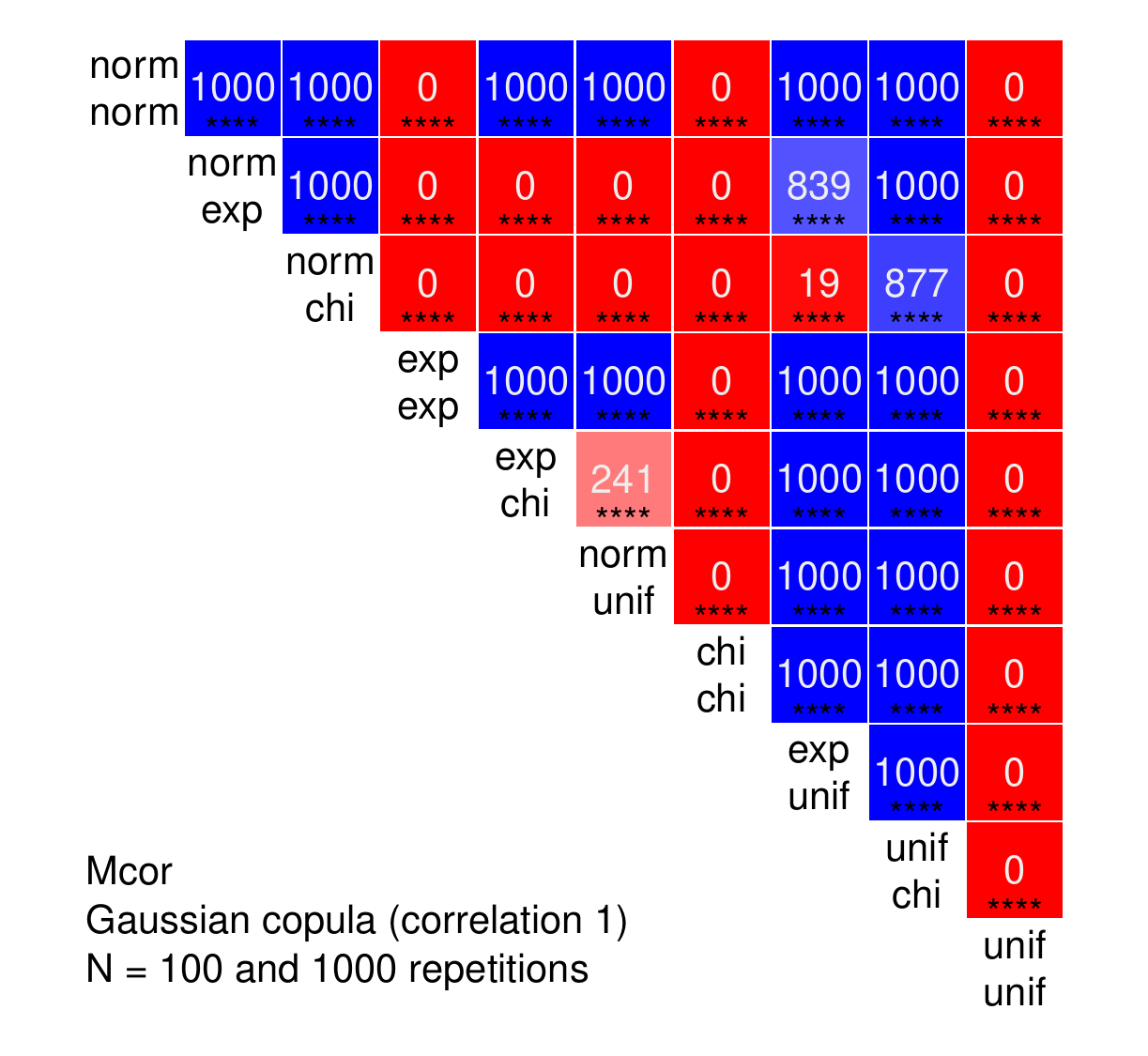}
\caption{Systematic dominance due to marginal distributions. Extension of Figure \ref{fig:pair-freq-unbiased} by considering 0, 0.1, 0.25, 0.5, 0.8 and 1 as parameters for the Gaussian copula. Observation:
The stronger the dependence is, the stronger the systematic problems appear.
}		\label{fig:pair-freq-unbiased-full}
\end{figure}

\begin{figure}[H]
%\newlength{\myl}
\centering
\setlength{\myl}{0.49\textwidth}
\includegraphics[width = \myl]{./Figs/pair-frequencies-ucmcor-1.pdf}
\includegraphics[width = \myl]{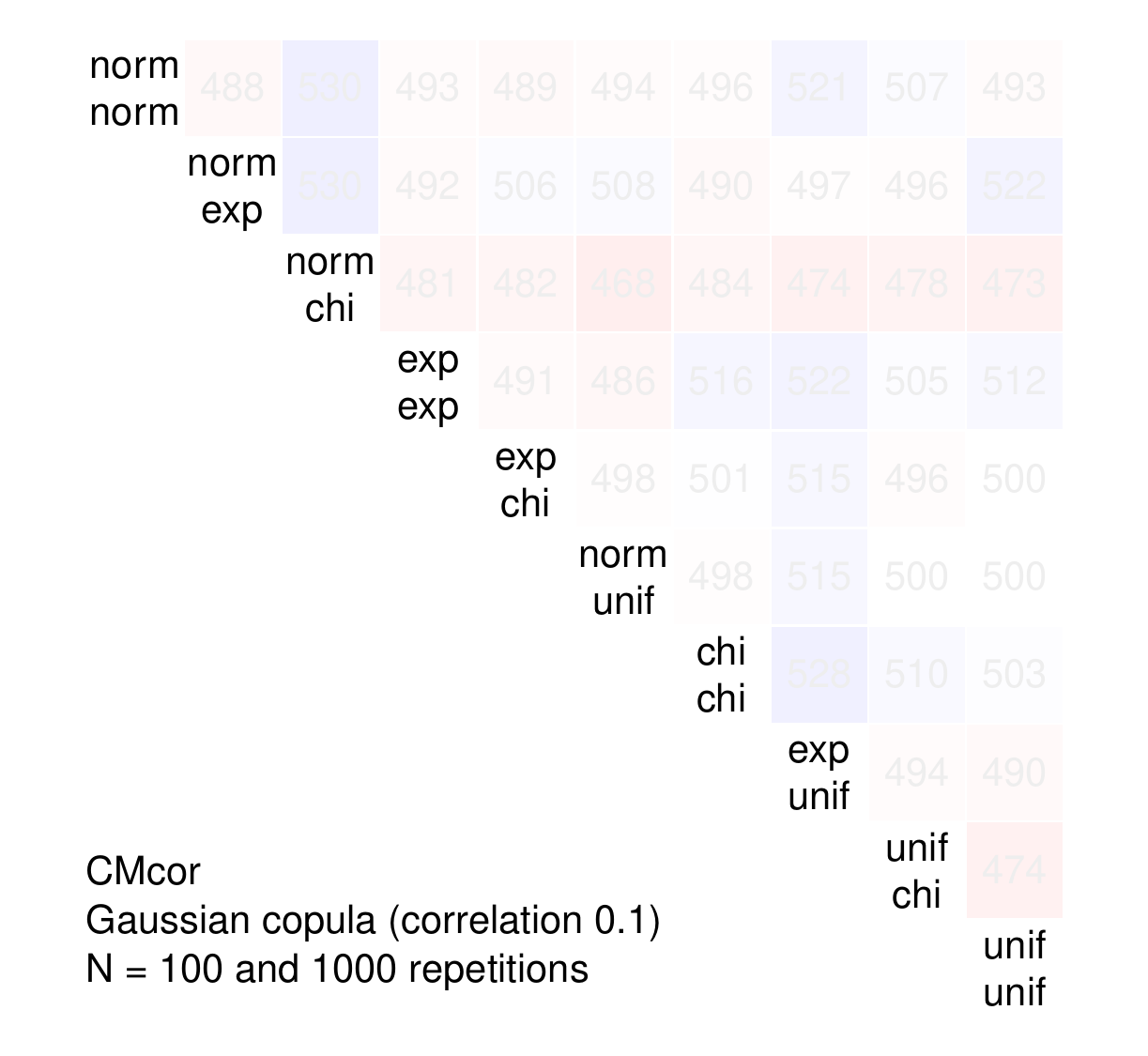}
\includegraphics[width = \myl]{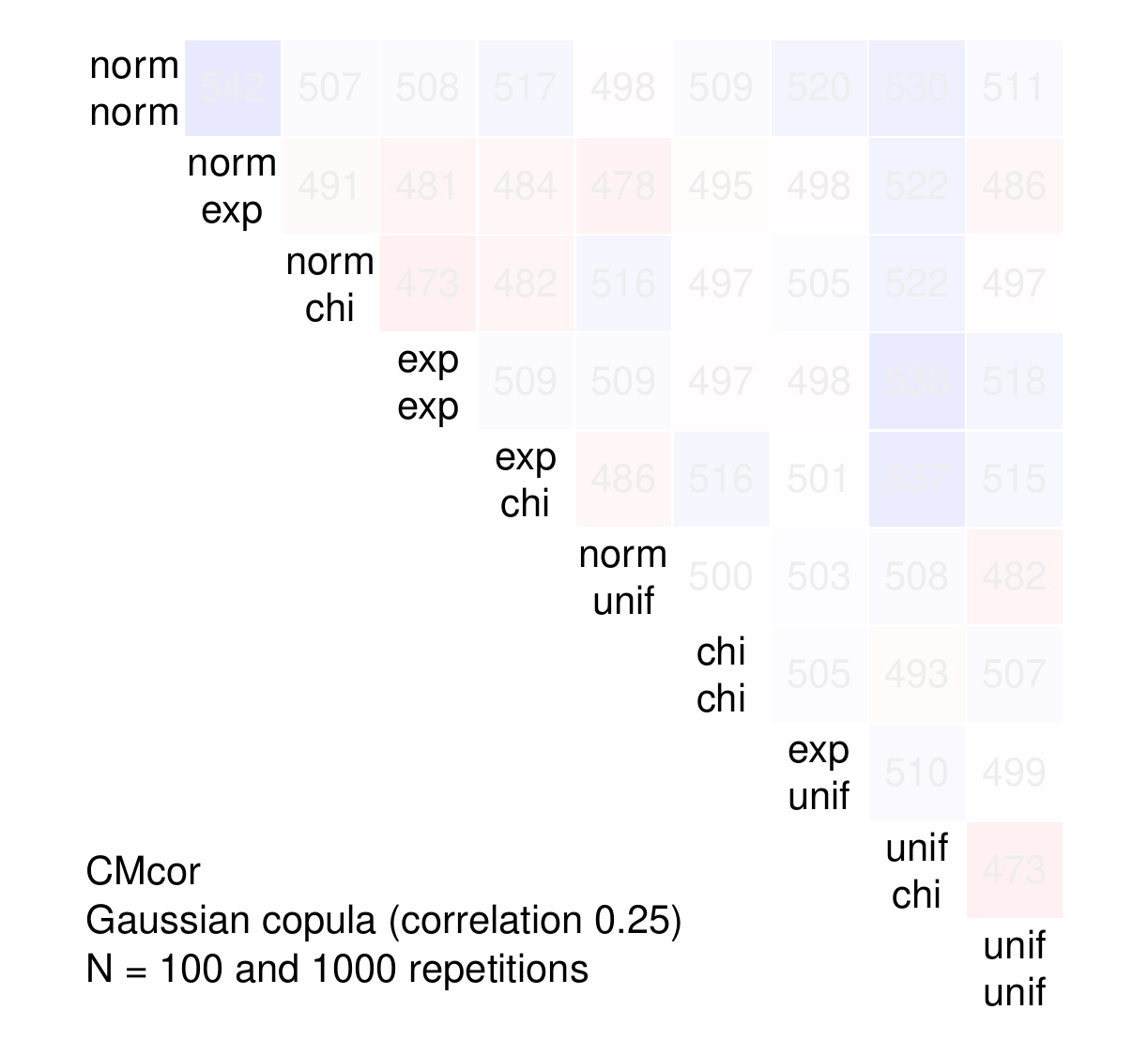} 
\includegraphics[width = \myl]{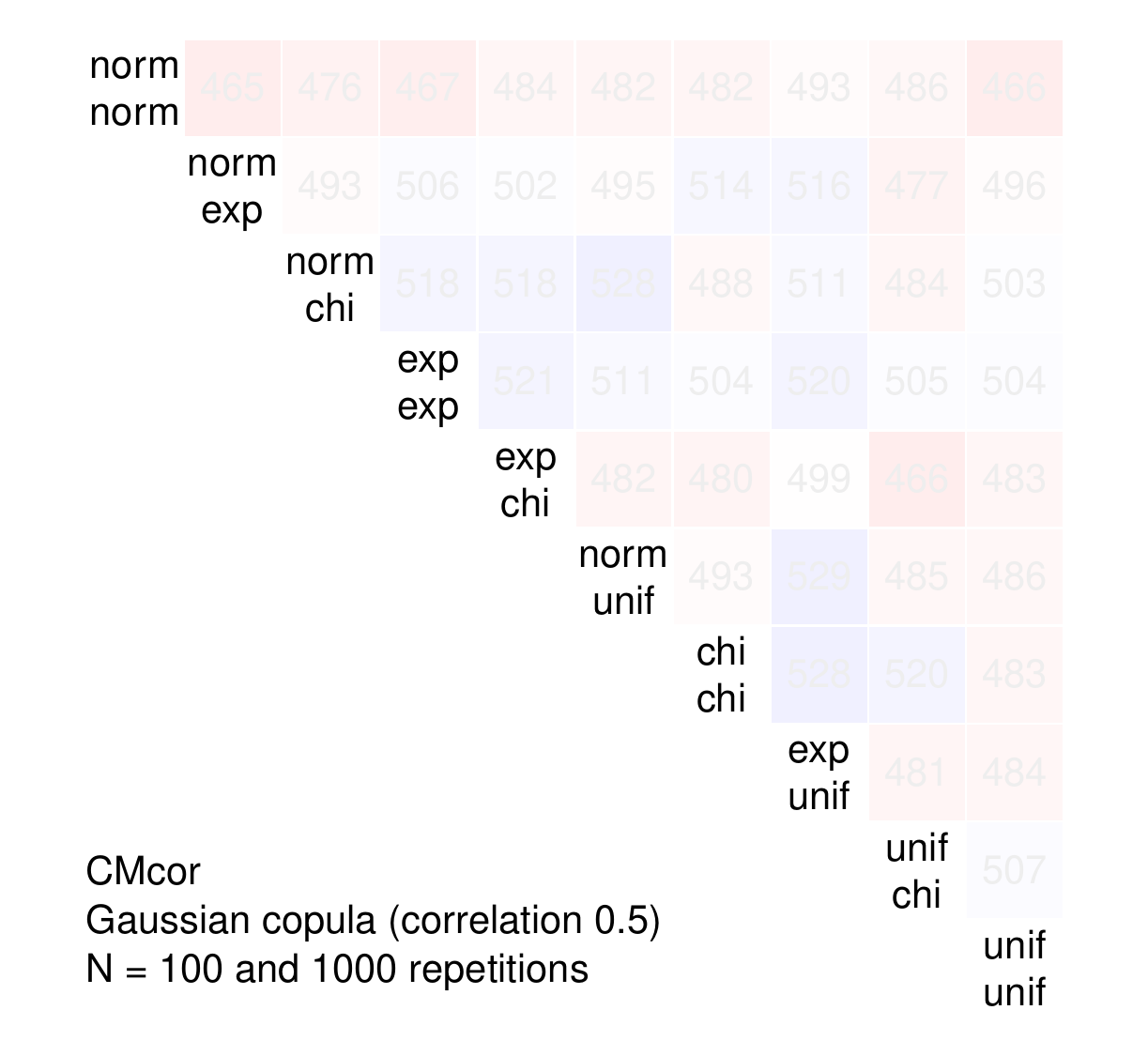}
\includegraphics[width = \myl]{./Figs/pair-frequencies-ucmcor-5.pdf}
\includegraphics[width = \myl]{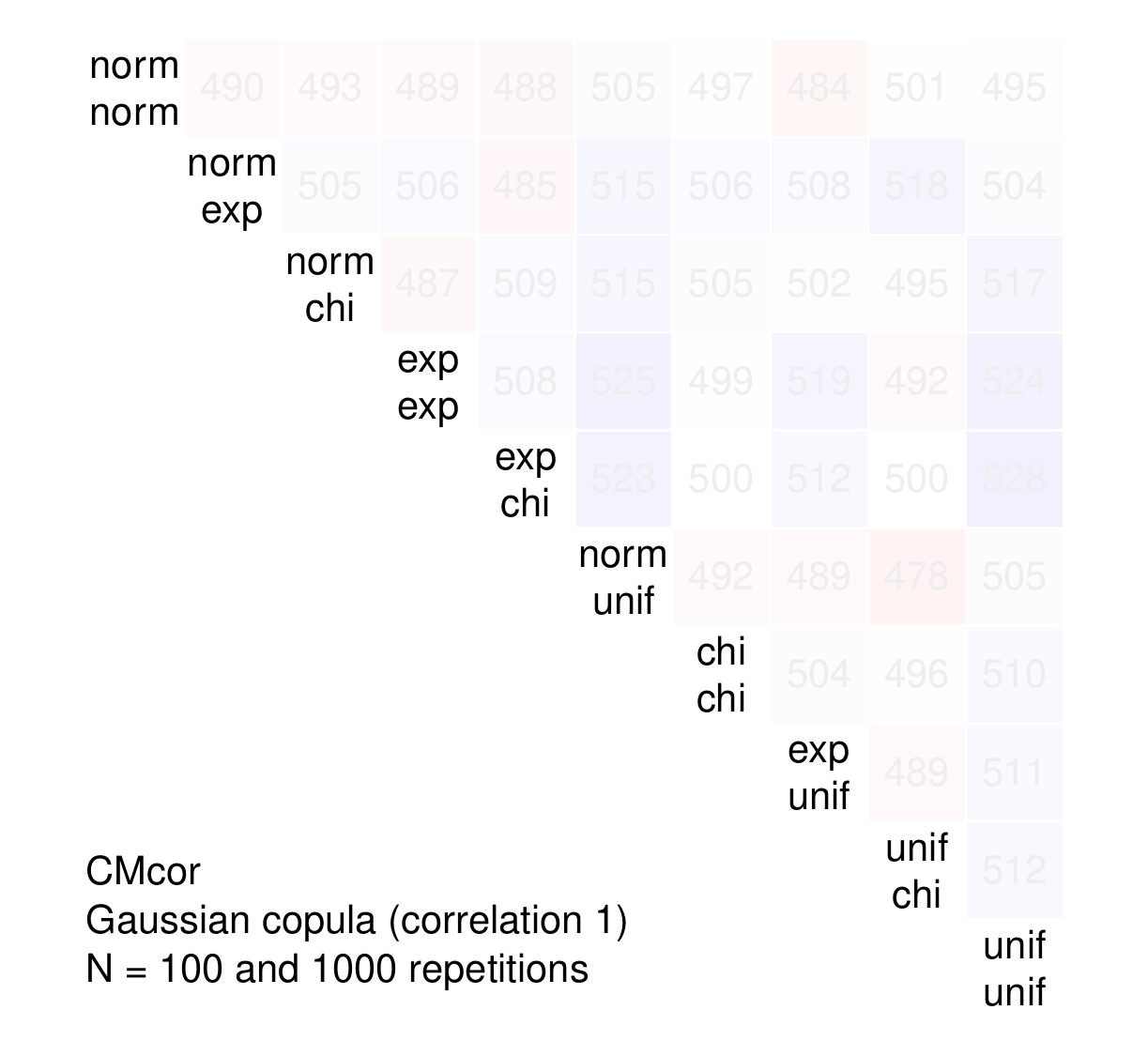}
\caption{Invariance with respect to marginal distributions. Extension of Figure \ref{fig:pair-freq-cop} by considering 0, 0.1, 0.25, 0.5, 0.8 and 1 as parameters for the Gaussian copula. Observation: There is no systematic dependence on the marginal distributions for the copula version of distance multicorrelation}		\label{fig:pair-freq-cop-full}
\end{figure}

\begin{figure}[H]
%\newlength{\myl}
\setlength{\myl}{1\textwidth}
\includegraphics[width = \myl]{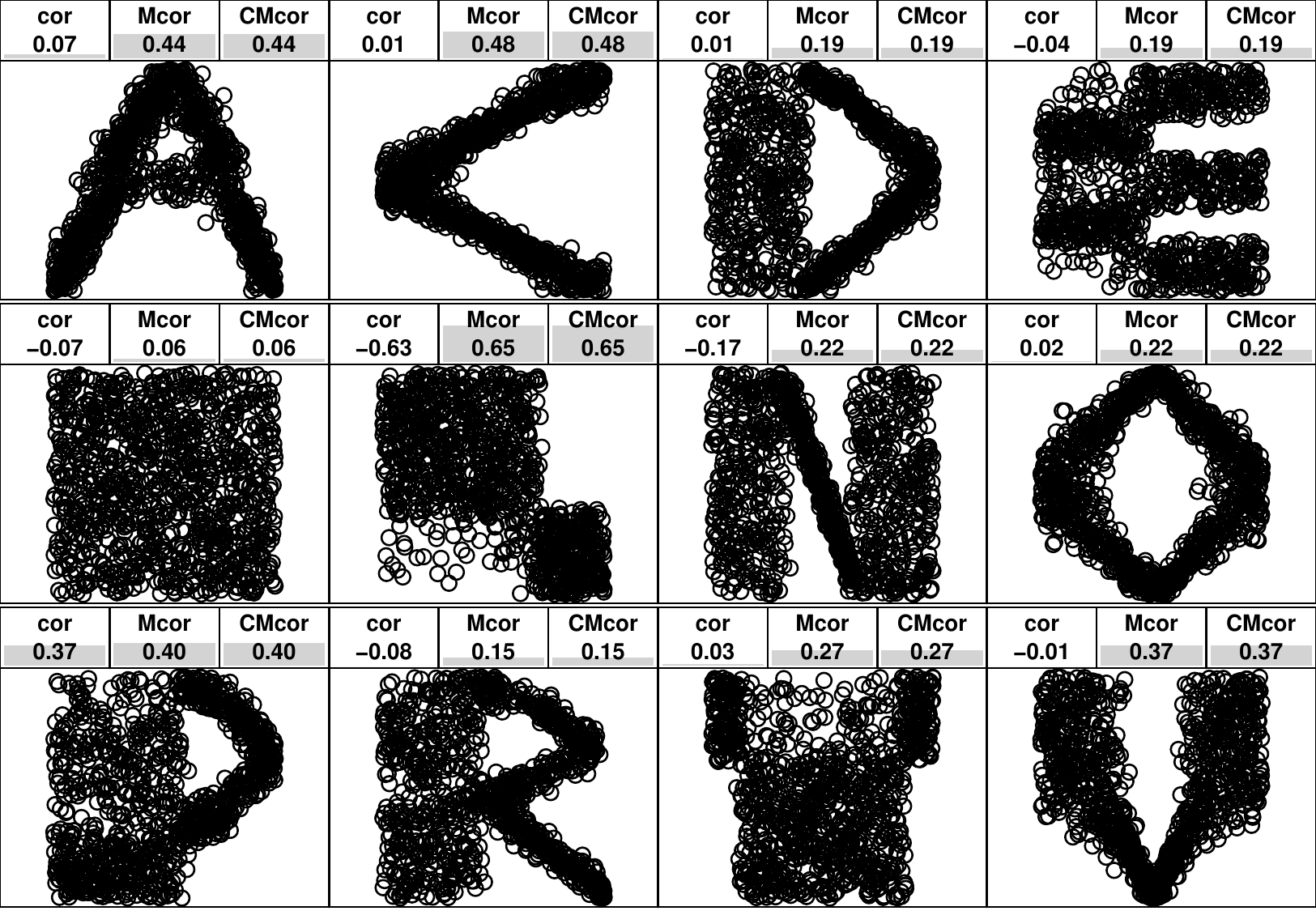}
\caption{Results of the Monte Carlo distributional transform. The datasets of Figure \ref{fig:letters} are depicted after the distributional transform has been applied, and the values of the measures are stated as in the other figures. Observation: Also for the transformed samples Pearson's correlation is not able to detect most dependencies. Moreover, by the transformation values of discontinuous random variables become uniformly distributed over rectangular areas.}		\label{fig:dist-trans-letter}
\end{figure}

\end{document}